\documentclass[11pt]{article}
\usepackage{amssymb}
\usepackage{amsthm}
\usepackage{amsmath}
\usepackage{fullpage,hyperref,mathrsfs,txfonts,epsfig,graphicx,graphics}
\newcommand{\fin}{\nolinebreak\hspace{\stretch{1}}$\lhd$}

\newcommand\remove[1]{}

\newcommand{\rnote}[1]{}
\newcommand{\jnote}[1]{}

\newcommand{\Rad}{\mathbf{\mathrm{\bf Rad}}}

\renewcommand{\L}{\mathscr{L}}

\newcommand{\1}{\mathbf{1}}
\newcommand{\e}{\varepsilon}
\newcommand{\R}{\mathbb{R}}
\newcommand{\E}{\mathbb{E}}
\newcommand{\N}{\mathbb{N}}

\newcommand{\F}{\mathscr{F}}

\newcommand{\Alg}{\mathrm{Alg}}

\newcommand{\SDP}{\mathrm{SDP}}

\newcommand{\Clust}{\mathrm{\bf Clust}}

\def\calG{{\cal G}}

\def\calL{{\cal L}}

\def\calX{{\cal X}}
\def\calX{{\cal X}}

\def\chop{ {\rm chop}}
\newcommand{\eps}{\varepsilon}
\newcommand{\infl}{{\rm Inf}}

\newtheorem{theorem}{Theorem}[section]
\newtheorem{lemma}[theorem]{Lemma}

\newtheorem{corollary}[theorem]{Corollary}

\newtheorem{fact}[theorem]{Fact}

\newtheorem{remark}{Remark}[section]

\newtheorem{conjecture}{Conjecture}
\newtheorem{question}[conjecture]{Question}

\date{}

\title{Sharp kernel clustering algorithms and their associated Grothendieck inequalities}

\author{Subhash Khot \and
Assaf Naor}
\begin{document}

\maketitle

\begin{abstract}
In the kernel clustering problem we are given a (large) $n\times n$
symmetric positive semidefinite matrix $A=(a_{ij})$ with
$\sum_{i=1}^n\sum_{j=1}^n a_{ij}=0$ and a (small) $k\times k$
symmetric positive semidefinite matrix $B=(b_{ij})$. The goal is to
find a partition $\{S_1,\ldots,S_k\}$ of $\{1,\ldots n\}$ which
maximizes  $ \sum_{i=1}^k\sum_{j=1}^k \left(\sum_{(p,q)\in S_i\times
S_j}a_{pq}\right)b_{ij}$.
 We design
a polynomial time approximation algorithm that achieves an
approximation ratio of $\frac{R(B)^2}{C(B)}$, where $R(B)$ and
$C(B)$ are geometric parameters that depend only on the matrix $B$,
defined as follows: if $b_{ij} = \langle v_i, v_j \rangle$ is the
Gram matrix representation of $B$ for some $v_1,\ldots,v_k\in \R^k$
then $R(B)$ is the minimum radius of a Euclidean ball containing the
points $\{v_1, \ldots, v_k\}$. The parameter $C(B)$ is defined as
the maximum over all measurable partitions $\{A_1,\ldots,A_k\}$ of
$\R^{k-1}$ of the quantity $\sum_{i=1}^k\sum_{j=1}^k b_{ij}\langle
z_i,z_j\rangle$, where for $i\in \{1,\ldots,k\}$ the vector $z_i\in
\R^{k-1}$ is the Gaussian moment of $A_i$, i.e.,
$z_i=\frac{1}{(2\pi)^{(k-1)/2}}\int_{A_i}xe^{-\|x\|_2^2/2}dx$. We
also show that for every $\eps
> 0$, achieving an approximation guarantee of $(1-\e)\frac{R(B)^2}{C(B)}$ is Unique
Games hard.
\end{abstract}

\section{Introduction}

Kernel Clustering~\cite{SSGB07} is a combinatorial optimization
problem which originates in the theory of machine learning. It is a
general framework for clustering massive statistical data so as to
uncover a certain hypothesized structure. The problem is defined as
follows: let $A=(a_{ij})$ be an $n\times n$ symmetric positive
semidefinite matrix which is usually normalized to be centered,
i.e., $\sum_{i=1}^n\sum_{j=1}^n a_{ij}=0$. The matrix $A$ is often
thought of as the correlation matrix of random variables
$(X_1,\ldots,X_n)$ that measure attributes of certain empirical
data, i.e., $a_{ij}=\E\left[X_iX_j\right]$. We are also given
another symmetric positive semidefinite $k\times k$ matrix
$B=(b_{ij})$ which functions as a hypothesis, or test matrix. Think
of $n$ as huge and $k$ as small. The goal is to cluster $A$ so as to
obtain a smaller matrix which most resembles $B$. Formally, we wish
to find a partition $\{S_1,\ldots,S_k\}$ of $\{1,\ldots,n\}$ so that
if we write $c_{ij}\coloneqq\sum_{(p,q)\in S_i\times S_j} a_{pq}$,
i.e., we form a $k\times k$ matrix $C=(c_{ij})$ by clustering $A$
according to the given partition, then the resulting clustered
version of $A$ has the maximum correlation $\sum_{i=1}^k\sum_{j=1}^k
c_{ij}b_{ij}$ with the hypothesis matrix $B$. Equivalently, the goal
is to evaluate the number:
\begin{equation}\label{eq:def clust}
\Clust(A|B)\coloneqq
\max_{\sigma:\{1,\ldots,n\}\to\{1,\ldots,k\}}\sum_{i=1}^k\sum_{j=1}^k
a_{ij}b_{\sigma(i)\sigma(j)}.
\end{equation}

The strength of this generic clustering framework is based in part
on the flexibility of adapting the matrix $B$ to the problem at
hand. Various particular choices of $B$ lead to well studied
optimization problems, while other specialized choices of $B$ are
based on statistical hypotheses which have been applied with some
empirical success. We refer to~\cite{SSGB07,KN08} for additional
background and a discussion of specific examples.

In~\cite{KN08} we investigated the computational complexity of the
kernel clustering problem. Answering a question posed
in~\cite{SSGB07}, we showed that this problem has a constant factor
polynomial time approximation algorithm. We refer to~\cite{KN08} for
more information on the best known approximation guarantees. We also
obtained hardness results for kernel clustering under various
complexity assumptions. For example, we showed in~\cite{KN08} that
when $B=I_3$ is the $3\times 3$ identity matrix then a
$\frac{16\pi}{27}$ approximation guarantee for $\Clust(A|I_3)$ is
achievable, while any approximation guarantee smaller than
$\frac{16\pi}{27}$ is Unique Games hard. We will discuss the Unique
Games Conjecture (UGC) presently. At this point it suffices to say
that the above statement is evidence that the hardness threshold of
the problem of approximating $\Clust(A|I_3)$ is $\frac{16\pi}{27}$,
or more modestly that obtaining a polynomial time algorithm which
approximates $\Clust(A|I_3)$ up to a factor smaller than
$\frac{16\pi}{27}$ would require a major breakthrough.

Another result proved in~\cite{KN08} is that when $k\ge 3$ and $B$
is either the $k\times k$ identity matrix or is spherical (i.e.,
$b_{ii}=1$ for all $i\in \{1,\ldots,k\}$) and centered (i.e.,
$\sum_{i=1}^k\sum_{j=1}^k b_{ij}=0$) then there is a polynomial time
approximation algorithm which, given $A$, approximates $\Clust(A|B)$
to within a factor of $\frac{8\pi}{9}\left(1-\frac{1}{k}\right)$. We
also presented in~\cite{KN08} a conjecture (called the Propeller
Conjecture) which we proved would imply that
$\frac{8\pi}{9}\left(1-\frac{1}{k}\right)$ is the UGC hardness
threshold when $B=I_k$. We refer to~\cite{KN08} for more information
on the Propeller Conjecture, which at present remains open.

The above quoted result from~\cite{KN08} settles the problem of
evaluating the UGC hardness threshold of the following type of
algorithmic task: given $A$ and an hypothesis matrix $B$ which is
guaranteed to belong to a certain class of matrices (in our case
centered and spherical), approximate efficiently the number
$\Clust(A|B)$. Naturally this can be refined to a family of
optimization problems which depend on a fixed $B$: for each $B$,
what is the UGC hardness threshold of the problem of, given $A$,
approximating $\Clust(A|B)$? In~\cite{KN08} we answered this
question only when $B=I_3$, and for $B=I_k$ assuming the Propeller
Conjecture, and asked about the case of general $B$ (we did give
some $B$-dependent bounds in~\cite{KN08}, but they were not sharp
for $B\neq I_k$ for reasons that will become clear presently). This
is a natural question since it makes sense to use the best possible
polynomial time algorithm if we know $B$ in advance.

Here we answer the above question in full generality. To explain our
results we need to define two geometric parameters which are
associated to $B$. Since $B$ is symmetric and positive semidefinite
we can find vectors $v_1,\ldots,v_k\in \R^k$ such that $B$ is their
Gram matrix, i.e.,  $b_{ij}=\langle v_i,v_j\rangle$ for all $i,j\in
\{1,\ldots,k\}$. Let $R(B)$ be the smallest possible radius of a
Euclidean ball in $\R^k$ which contains $\{v_1,\ldots,v_k\}$ and let
$w(B)$ be the center of this ball. Let $C(B)$ be the maximum over
all partitions $\{A_1,\ldots,A_k\}$ of $\R^{k-1}$ into measurable
sets of the quantity $\sum_{i=1}^k\sum_{j=1}^k b_{ij}\langle
z_i,z_j\rangle$, where for $i\in \{1,\ldots,k\}$ the vector $z_i\in
\R^{k-1}$ is the Gaussian moment of $A_i$, i.e.,
$z_i=\frac{1}{(2\pi)^{(k-1)/2}}\int_{A_i}xe^{-\|x\|_2^2/2}dx$ (this
maximum exists, as shown in Section~\ref{sec:simplices}). Our main
result is the following theorem\footnote{We refer to the discussion
in Question 1 in Section~\ref{sec:open} below which addresses the
issue of computing efficiently good approximate clusterings rather
than approximating only the value $\Clust(A|B)$.}:

\begin{theorem}\label{thm:main intro}
For every symmetric positive semidefinite $k\times k$ matrix $B$
there exists a randomized polynomial time algorithm which given an
$n\times n$ symmetric positive semidefinite centered matrix $A$,
outputs a number $\Alg(A)$ such that
$$
\Clust(A|B)\le \E\left[\Alg(A)\right]\le
\frac{R(B)^2}{C(B)}\Clust(A|B).
$$

On the other hand, assuming the Unique Games Conjecture, no
polynomial time algorithm approximates $\Clust(A|B)$ to within a
factor strictly smaller than $\frac{R(B)^2}{C(B)}$.
\end{theorem}

As an example of Theorem~\ref{thm:main intro} for a particular
hypothesis matrix consider the following perturbation of the
previously studied case $B=I_3$: $$ B_c\coloneqq
\begin{pmatrix}
  1 & 0 & 0 \\
   0 & 1 & 0\\
   0 & 0 & c
   \end{pmatrix},$$ where $c>0$ is a parameter. The problem of
   approximating efficiently $\Clust(A|B_c)$ corresponds to
   partitioning the rows of $A$ into $3$ sets $S_1,S_2,S_3\subseteq \{1,\ldots,n\}$ and maximizing the sum of
   the total masses of $A$ on $S_1\times S_1,S_2\times S_2,S_3\times
   S_3$, where the parameter $c$ can be used to tune the weight of
   the set $S_3$. This problem is not particularly important---we
   chose it just as a concrete example for the sake of illustration.
   In Section~\ref{sec:example} we compute the parameters
   $R(B_c),C(B_c)$ and deduce that the UGC hardness threshold of the
   problem of computing $\Clust(A|B_c)$ equals $\frac{4\pi
c(1+c)^2}{(1+2c)^3}$ if $c\ge\frac12$ and equals $\frac{\pi
(1+c)^2}{2+4c}$ if $c\le\frac12$. The change at $c=\frac12$
corresponds in a qualitative change in the best algorithm for
computing $\Clust(A|B_c)$---we refer to Section~\ref{sec:example}
for an explanation.

In the remainder of this introduction we will explain the various
ingredients of Theorem~\ref{thm:main intro} (in particular the
Unique Games Conjecture), and the new ideas used in its proof.

The main tool in the design of the algorithm in
Theorem~\ref{thm:main intro} is a natural generalization of the
positive semidefinite Grothendieck inequality. In~\cite{Gro53}
Grothendieck proved that there exists a universal constant $K>0$
such that for every $n\times n$ symmetric positive semidefinite
matrix $A=(a_{ij})$ we have\footnote{This inequality is sometimes
written as $ \max_{x_i,y_i\in S^{n-1}}\sum_{i=1}^n\sum_{j=1}^n
a_{ij} \langle x_i,y_j\rangle \le K\max_{\e_i,\delta_i\in
\{-1,1\}}\sum_{i=1}^n\sum_{j=1}^n a_{ij}\e_i\delta_j$, but it is
easy (and standard) to verify that since $A$ is positive
semidefinite this formulation coincides with~\eqref{eq:PSD gro}.}:
\begin{equation}\label{eq:PSD gro}
\max_{x_1,\ldots,x_n\in S^{n-1}}\sum_{i=1}^n\sum_{j=1}^n a_{ij}
\langle x_i,x_j\rangle \le K\max_{\e_1,\ldots,\e_n\in
\{-1,1\}}\sum_{i=1}^n\sum_{j=1}^n a_{ij}\e_i\e_j.
\end{equation}
The best constant $K$ in~\eqref{eq:PSD gro} was shown
in~\cite{Rie74} to be equal to $\frac{\pi}{2}$. A natural variant
of~\eqref{eq:PSD gro}  is to replace the numbers $-1,1$ by general
$v_1,\ldots,v_k\in \R^k$, namely one might ask for the smallest
constant $K>0$ such that for every symmetric positive semidefinite
$n\times n$ matrix $A$ we have:
\begin{equation}\label{eq:our gro with K intro}
\max_{x_1,\ldots,x_n\in S^{n-1}}\sum_{i=1}^n\sum_{j=1}^n a_{ij}
\langle x_i,x_j\rangle \le K\max_{u_1,\ldots,u_n\in
\{v_1,\ldots,v_k\}}\sum_{i=1}^n\sum_{j=1}^n a_{ij}\langle
u_i,u_j\rangle.
\end{equation}
In Section~\ref{sec:ineq} we prove that~\eqref{eq:our gro with K
intro} holds with $K=\frac{1}{C(B)}$, where $B=\left(\langle
v_i,v_j\rangle\right)$ is the Gram matrix of $v_1,\ldots,v_k$, and
that this constant is sharp. This inequality is proved along the
following lines. Fix $n$ unit vectors $x_1,\ldots,x_n\in S^{n-1}$.
Let $G=(g_{ij})$ be a $(k-1)\times n$ random matrix whose entries
are i.i.d. standard Gaussian random variables. Let $A_1,\ldots,A_k\
\subseteq \R^{k-1}$ be a measurable partition of $\R^{k-1}$ at which
$C(B)$ is attained. Define a random choice of $u_i\in
\{v_1,\ldots,v_k\}$ by setting $u_i=v_\ell$ for the unique $\ell\in
\{1,\ldots,k\}$ such that $Gx_i\in A_\ell$. The fact
that~\eqref{eq:our gro with K intro} holds with $K=\frac{1}{C(B)}$
is a consequence of the following fact, which we prove in
Section~\ref{sec:ineq}:
\begin{equation}\label{eq:in expectation}
\E\left[\sum_{i=1}^n\sum_{j=1}^n a_{ij}\langle
u_i,u_j\rangle\right]\ge C(B)\sum_{i=1}^n\sum_{j=1}^n a_{ij} \langle
x_i,x_j\rangle.
\end{equation}
The crucial point in the proof of~\eqref{eq:in expectation} is the
following identity, proved in Lemma~\ref{lem:hermite} as a corollary
of the closed-form formula for the Poison kernel of the Hermite
polynomials: for every two measurable subsets $E,F\subseteq
\R^{k-1}$ and any two unit vectors $x,y\in \R^n$, we have
\begin{multline}\label{poisson identity intro}
\Pr\left[Gx\in E\ \mathrm{and}\ Gy\in
F\right]\\=\gamma_{k-1}(E)\gamma_{k-1}(F)+\langle
x,y\rangle\left\langle\int_Eud\gamma_{k-1}(u),\int_Fud\gamma_{k-1}(u)\right\rangle+\sum_{\ell=2}^\infty
\left\langle x^{\otimes \ell},y^{\otimes
\ell}\right\rangle\sum_{\substack{s\in
(\N\cup\{0\})^{k-1}\\s_1+\cdots+s_{k-1}=\ell}}\alpha_s(E)\alpha_s(F),
\end{multline}
for some real coefficients $\{\alpha_s(E)\}_{s\in
(\N\cup\{0\})^{k-1}},\{\alpha_s(F)\}_{s\in
(\N\cup\{0\})^{k-1}}\subseteq \R$. Here $\gamma_{k-1}$ denotes the
standard Gaussian measure on $\R^{k-1}$. The product structure of
the decomposition~\eqref{poisson identity intro} hints at the role
of the fact that $A$ is positive semidefinite in the proof
of~\eqref{eq:in expectation}---the complete details appear in
Section~\ref{sec:ineq}.

Once the generalized Grothendieck inequality~\eqref{eq:our gro} is
obtained with $K=\frac{1}{C(B)}$ it is simple to design the
algorithm whose existence is claimed in Theorem~\ref{thm:main
intro}, which is based on semidefinite programming---this is done in
Section~\ref{sec:alg}.

We shall now pass to an explanation of the hardness result in
Theorem~\ref{thm:main intro}. The Unique Games Conjecture, posed by
Khot in~\cite{Khot02}, is as follows. A { Unique Game} is an
optimization problem with an instance ${\L}= {\L}( G(V,W,E), n,
\{\pi_{vw}\}_{(v,w) \in E})$. Here $G(V,W,E)$ is a regular bipartite
graph with vertex sets $V$ and $W$ and edge set $E$. Each vertex is
supposed to receive a label from the set $\{1,\ldots, n\}$.
For every edge $(v,w) \in E$ with $v \in V$ and $w\in W$, there is a
given permutation $\pi_{vw}: \{1,\ldots,n\} \to \{1,\ldots, n\}$. A
labeling of the Unique Game instance is an assignment $ \rho: V \cup
W \to \{1,\ldots, n\}$. An edge $(v,w)$ is satisfied by a labeling
$\rho$ if and only if $\rho(v) = \pi_{vw}(\rho(w))$. The goal is to
find a labeling that maximizes the fraction of edges satisfied (call
this maximum ${\rm OPT}({\L}))$. We think of the number of labels
$n$ as a constant and the size of the graph $G(V,W,E)$ as the size
of the problem instance. The Unique Games Conjecture (UGC) asserts
that for  arbitrarily small constants $\e, \delta > 0$, there exists
a constant $n = n(\e, \delta)$ such that no polynomial time
algorithm can distinguish whether a Unique Games instance ${\L}=
{\L}( G(V,W,E), n, \{\pi_{vw}\}_{(v,w) \in W})$  satisfies ${\rm
OPT}({\L}) \leq \delta$ (soundness) or there exists a labeling such
that for $1-\e$ fraction of the vertices $v\in V$ all the edges
incident with $v$ are satisfied (completeness)\footnote{This version
of the UGC is not the standard version  as stated in~\cite{Khot02},
which only requires ${\rm OPT}({\L}) \geq 1-\eps$ in the
completeness. However, it was shown in~\cite{KR08} that this
seemingly stronger version of the UGC actually follows from the
original UGC---we will require this stronger statement in our
proofs.}.
 This
conjecture is (by now) a commonly used complexity assumption to
prove hardness of approximation results. Despite several recent
attempts to get better polynomial time approximation algorithms  for
the Unique Game problem (see the table in~\cite{CMM06} for a
description of known results), the unique games conjecture still
stands.

Our UGC hardness result follows the standard ``dictatorship test"
approach which is prevalent in PCP based hardness proofs, with a new
twist which seems to be of independent interest. Since the kernel
clustering problem is concerned with an assignment of one of $k$
labels to each of the rows of the matrix $A$, the natural setting of
our hardness proof is a dictatorship test for functions on
$\{1,\ldots,k\}^n$ taking values in $\{1,\ldots,k\}$ (this was
already the case in~\cite{KN08}). The general ``philosophy" of such
hardness proofs is to associate to every such function a certain
numerical parameter called the ``objective value" (which is adapted
to the optimization problem at hand). The general scheme is to show
that for some numbers $a,b>0$, if $f$ depends on only one coordinate
(i.e., it is a ``dictatorship") then the objective value of $f$ is
at least $a$, while if $f$ does not have any coordinate which is too
influential then the objective value of $f$ is at most $b+o(1)$ (the
$o(1)$ depends on the notion of having no influential coordinates
and its exact form is not important for the purpose of this
overview---we refer to Section~\ref{sec:UGC} for details). Once such
a result is proved, techniques from the theory of Probabilistically
Checkable Proofs can show that under a suitable complexity theoretic
assumption (in our case the UGC) no polynomial time algorithm can
achieve an approximation factor smaller than $\frac{a}{b}$.

Implicit to the above discussion is an underlying product
distribution on $\{1,\ldots,k\}^n$ with respect to which we measure
the influence of variables. In~\cite{KN08} the case of $B=I_k$ was
solved using the uniform distribution on $\{1,\ldots,k\}$. It turns
out that in order to prove the sharp hardness result in
Theorem~\ref{thm:main intro} we need to use a non-uniform
distribution which depends on the geometry of $B$. Namely, writing
$B$ as a Gram matrix $b_{ij}=\langle v_i,v_j\rangle$, recall that
$R(B)$ is the radius of the smallest Euclidean ball containing
$\{v_1,\ldots,v_k\}$ and $w(B)$ is the center of this ball. A simple
separation argument shows that $w(B)$ is in the convex hull of the
vectors in $\{v_1,\ldots,v_k\}$ whose distance from $w(B)$ is
exactly $R(B)$. Writing $w(B)$ as a convex combination of these
points and considering the coefficients of this convex combination
results in a probability distribution on $\{1,\ldots,k\}$. In our
hardness proof we use the $n$-fold product of (a small perturbation
of) this probability distribution as the underlying distribution on
$\{1,\ldots,k\}$ for our dictatorship test---see Figure 1 for a
schematic description of the situation described above. The full
details of this approach, including all the relevant definitions,
are presented in Section~\ref{sec:UGC}.

\begin{figure}[h]\label{fig:subu}
\begin{center}\includegraphics[scale=0.377]{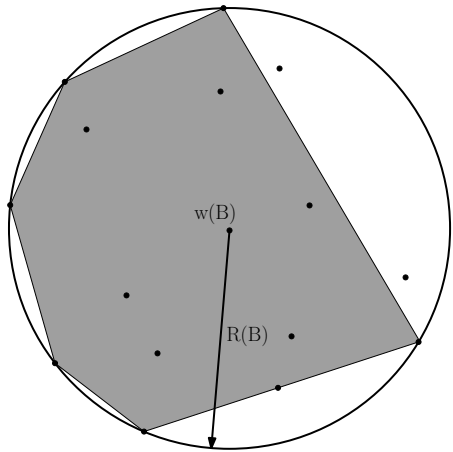}
\end{center}
 \caption{{{\em  The geometry of the test matrix $B$ induces a dictatorship test: the points above are the vectors
 $\{v_1,\ldots,v_k\}\subseteq \R^k$ such that $B$ is their Gram matrix.
 The ball depicted above  is the smallest
 Euclidean ball containing $\{v_1,\ldots,v_k\}$, $R(B)$ is its radius and $w(B)$ is its center. Then $w(B)$ is in the convex hull
 of the points in $\{v_1,\ldots,v_k\}$ which are at distance exactly $R(B)$ from $w(B)$. Writing $w(B)$ as a convex combination of
 these boundary points yields a distribution over the labels $\{1,\ldots,k\}$. Our dictatorship test corresponds to selecting
 a point from the $n$-fold power of this probability space and comparing the behavior
 of a certain ``objective value" (defined in equation~\eqref{eq:OBJ} below), which depends only on the singleton
 Fourier coefficients, for dictatorships and for functions with low influences.       }}} \label{fig:hull}
\end{figure}

\subsection{Open problems}\label{sec:open}

We end this introduction with a statement of some open problems.

\begin{question}\label{Q:
fixed parameter} {\em Theorem~\ref{thm:main intro} shows that the
UGC hardness threshold of the problem of computing $\Clust(A|B)$ for
a fixed hypothesis matrix $B$ equals $\frac{R(B)^2}{C(B)}$. It is
natural to ask if there is also a polynomial time algorithm which
outputs a clustering of $A$ whose value is within a factor of
$\frac{R(B)^2}{C(B)}$ of the optimal clustering. The issue is that
our rounding algorithm uses the partition $\{A_1,\ldots,A_k\}$ of
$\R^{k-1}$ at which $C(B)$ is attained. In
Section~\ref{sec:simplices} we study this optimal partition, and
show that it has a relatively simple structure rather than being
composed of general measurable sets: it corresponds to cones which
are induced by the faces of a simplex. This information allows us to
compute efficiently a partition which comes as close as we wish to
the optimal partition when $k$ is fixed, or grows slowly with $n$
(to be safe lets just say for the sake of argument that $k\approx
\log\log n$ works). We refer to Remark~\ref{rem:how to compute} for
details. We currently do not know if there is polynomial time
rounding algorithm when, say, $k\approx \sqrt{n}$. Given $\e>0$, is
there an algorithm which, given $A$ and $B$, computes $\Clust(A|B)$
to within a factor of $(1+\e)\frac{R(B)^2}{C(B)}$, and runs in time
which is polynomial in both $n$ and $k$ (and maybe even $1/\e$)?}
\end{question}

\begin{question}\label{Q:propeller}
{\em We remind the reader that the Propeller Conjecture remains
open. This conjecture is about the value of $C(I_k)$ when $k\ge 4$.
It states that the partition at which $C(I_k)$ is attained is
actually much simpler than what one might initially expect: only $3$
of the sets have positive measure and they form a cylinder over a
planar $120^\circ$ ``propeller". We refer to~\cite{KN08} for a
precise formulation and some evidence for the validity of the
Propeller Conjecture.}
\end{question}

\begin{question}\label{Q:not centered}
{\em The kernel clustering problem was stated in~\cite{SSGB07} for
matrices $A$ which are centered. This makes sense from the
perspective of machine learning, but it seems meaningful to also ask
for the UGC hardness threshold of the same problem when $A$ is not
assumed to be centered. In the present paper we did not investigate
this case at all, and it seems that the exact UGC hardness threshold
when $A$ is not necessarily centered is not known for any
interesting hypothesis matrix $B$. Note that in~\cite{KN08} we
showed that there is a constant factor polynomial time approximation
algorithm when $A$ is not necessarily centered: we obtained
in~\cite{KN08} an approximation guarantee of $1+\frac{3\pi}{2}$ in
this case, but this is probably suboptimal.}
\end{question}

\section{Preliminaries on the parameter $C(B)$}\label{sec:simplices}

Let $B=(b_{ij})_{i,j=1}^k\in M_k(\R)$ be a $k\times k$ symmetric
positive semidefinite matrix. In what follows we fix $k\ge 2$ and
the matrix $B$. We also fix vectors $v_1,\ldots,v_k\in \R^k$ for
which $b_{ij}=\langle v_i,v_j\rangle$ for all $i,j\in
\{1,\ldots,k\}$.

Let $\gamma_n$ denote the standard Gaussian measure on $\R^n$, i.e.,
the density of $\gamma_n$ is
$\frac{1}{(2\pi)^{n/2}}e^{-\|x\|_2^2/2}$. We denote by $H_k$ the
Hilbert space $L_2(\gamma_n)\oplus L_2(\gamma_n)\oplus \cdots\oplus
L_2(\gamma_n)$ ($k$ times) and we consider the convex subset
$\Delta_k(\gamma_n)\subseteq H_k$ give by:
\begin{equation}\label{eq:def simplex}
\Delta_k(\gamma_n)\coloneqq \left\{(f_1,\ldots,f_k)\in H_k:\ \forall
j\in \{1,\ldots,k\}\ f_j\ge 0\ \wedge\ \sum_{j=1}^kf_j= 1\right\}.
\end{equation}

Define: \begin{equation}\label{eq:def cnB}
 C(n,B)\coloneqq \sup_{(f_1,\ldots,f_k)\in \Delta_k(\gamma_n)}
 \sum_{i=1}^k\sum_{j=1}^k b_{ij} \cdot
 \left\langle \int_{\R^n}xf_i(x)d\gamma_n(x), \int_{\R^n}xf_j(x)d\gamma_n(x)
 \right\rangle.
\end{equation}

The following lemma is a variant of Lemma 3.1 in~\cite{KN08} (but
see Remark~\ref{remark:simplex with 0} for an explanation of a
subtle difference). It simply states that the supremum
in~\eqref{eq:def cnB} is attained at a $k$-tuple of functions which
correspond to a partition of $\R^n$.

\begin{lemma}\label{lemma:cnb-partition}
  \label{lem:maximizer exists} There exist disjoint measurable sets $A_1,\ldots,A_k\subseteq
 \R^n$ such that $A_1\cup A_2\cup\cdots \cup A_k=\R^n$ and
$$
\sum_{i=1}^k\sum_{j=1}^k b_{ij} \cdot \left\langle
\int_{A_j}xd\gamma_n(x) , \int_{A_j}xd\gamma_n(x) \right\rangle
=C(n,B).
$$
 \end{lemma}

\begin{proof} Define $\Psi:\Delta_k(\gamma_n)\to \R$ by
\begin{equation}\label{eq:def psi}
\Psi(f_1,\ldots,f_k)\coloneqq \sum_{i=1}^k\sum_{j=1}^k b_{ij} \cdot
 \left\langle \int_{\R^n}xf_i(x)d\gamma_n(x), \int_{\R^n}xf_j(x)d\gamma_n(x)
 \right\rangle.
\end{equation}
We first observe that $\Psi$ is a convex function. Indeed, fix
$\lambda\in [0,1]$ and $(f_1,\ldots,f_k),(g_1,\ldots,g_k)\in
\Delta_k(\gamma_n)$. Denote $z_i\coloneqq \int_{\R^n} x
f_i(x)d\gamma_n(x)$ and $w_i\coloneqq \int_{\R^n} x
g_i(x)d\gamma_n(x)$ for every $i\in \{1,\ldots,k\}$. Then:
\begin{eqnarray*}
&&\!\!\!\!\!\!\!\!\!\!\!\!\!\!\!\!\lambda\Psi(f_1,\ldots,f_k)+(1-\lambda)\Psi(g_1,\ldots,g_k)-\Psi(\lambda
f_1+(1-\lambda)g_1,\ldots,\lambda f_k+(1-\lambda)g_k)\\&=&
\sum_{i=1}^k\sum_{j=1}^k\langle v_i,v_j\rangle\left(\lambda\langle
z_i,z_j\rangle+(1-\lambda)\langle w_i,w_j\rangle-\langle \lambda
z_i+(1-\lambda)w_i,\lambda z_j+(1-\lambda)w_j\rangle\right)\\
&=& \lambda(1-\lambda)\sum_{i=1}^k\sum_{j=1}^k\langle
v_i,v_j\rangle\langle z_i-w_i,z_j-w_j\rangle\\
&=& \lambda(1-\lambda)\left\|\sum_{i=1}^n v_i\otimes
(z_i-w_i)\right\|_2^2\ge 0.
\end{eqnarray*}

Since $\Delta_k(\gamma_n)$ is a weakly compact subset of $H_k$ and
$\Psi$ is weakly continuous and convex, $\Psi$ attains its maximum
(which equals $C(n,B)$) on $\Delta_k(\gamma_n)$ at an extreme point
of $\Delta_k(\gamma_n)$, say at $(f_1^*,\ldots,f_k^*)\in
\Delta_k(\gamma_n)$. It follows that there exist measurable sets
$A_1,\ldots,A_k\subseteq \R^n$ which form a partition of $\R^n$ such
that $(f_1^*,\ldots,f_k^*)=(\1_{A_1},\ldots,\1_{A_k})$ almost
everywhere\footnote{To see this standard fact observe that otherwise
there would be some $A\subseteq \R^n$ of positive measure, $\e\in
(0,1/2)$, and distinct $i,j\in\{1,\ldots,k\}$ such that
$f_i\1_{A},f_j\1_{A}\in (\e,1-\e)$. But $(f_1^*,\ldots,f_k^*)$ would
then not be an extreme point since it is the average of
$(g_1,\ldots,g_k),(h_1,\ldots,h_k)\in
\Delta_k(\gamma_n)\setminus\{(f_1^*,\ldots,f_k^*)\}$, where
$g_\ell=h_\ell=f_\ell^*$ for $\ell\in \{1,\ldots,k\}\setminus
\{i,j\}$ and $g_i=(f_i^*+\e)\1_A+f_i^*\1_{\R^n\setminus A}$,
$h_i=(f_i^*-\e)\1_A+f_i^*\1_{\R^n\setminus A}$,
$g_j=(f_j^*-\e)\1_A+f_j^*\1_{\R^n\setminus A}$,
$h_j=(f_j^*+\e)\1_A+f_j^*\1_{\R^n\setminus A}$.}, as
required.\end{proof}

\begin{remark}\label{remark:simplex with 0} {\em In~\cite{KN08} a stronger
result was proved when $B=I_k$ (the $k\times k$ identity matrix).
Namely, using the notation of the proof of
Lemma~\ref{lemma:cnb-partition} it was shown that the maximum of
$\Psi$ on the larger convex set
$$
\widetilde{\Delta_k(\gamma_n)}\coloneqq \left\{(f_1,\ldots,f_k)\in
H_k:\ \forall j\in \{1,\ldots,k\}\ f_j\ge 0\ \wedge\
\sum_{j=1}^kf_j\le 1\right\}
$$
is also attained at
$(f_1^*,\ldots,f_k^*)=(\1_{A_1},\ldots,\1_{A_k})$ for some
measurable sets $A_1,\ldots,A_k\subseteq \R^n$ which form a
partition of $\R^n$. It turns out that this stronger fact helps to
slightly simplify the proof of the corresponding UGC hardness
result. However, we do not know how to prove this stronger statement
for general $B$, so we formulated the weaker statement in
Lemma~\ref{lemma:cnb-partition}, at the cost of needing to modify
our proof of the UGC hardness result for general $B$ in
Section~\ref{sec:UGC}.

The same extreme point argument as in the proof of
Lemma~\ref{lemma:cnb-partition} shows that the maximum of $\Psi$ on
$\widetilde{\Delta_k(\gamma_n)}$ is attained at
$(f_1^*,\ldots,f_k^*)=(\1_{A_1},\ldots,\1_{A_k})$ for some disjoint
measurable sets $A_1,\ldots,A_k\subseteq \R^n$, but now it does not
follow that they necessarily cover all of $\R^n$. When $B=I_k$ it
can be shown as in~\cite{KN08} that these sets do cover $\R^n$. The
same statement is true when $B$ is diagonal, as we now show by
arguing as in the proof in~\cite{KN08}, but we do not know if it is
true for general $B$. So, assume that $B$ is diagonal with positive
diagonal entries $(b_1,\ldots,b_k)$. Let $A=\R^n\setminus
\bigcup_{i=1}^k A_k$. Denote $z_j\coloneqq \int_{A_j} x
d\gamma_n(x)$ and $w=\int_Ax d\gamma_n(x)$. Note that
$w+z_1+\cdots+z_k=0$. If $w=0$ then $\Psi$ attains its maximum on
the partition $\{A\cup A_1,A_2,\ldots,A_k\}$, so assume for the sake
of contradiction that $w\neq 0$. For every $i\in \{1,\ldots,k\}$ we
have:
\begin{multline*}
\sum_{j=1}^n b_j\|z_j\|_2^2=\Psi(\1_{A_1},\ldots,\1_{A_k})\ge
\Psi(\1_{A_1},\ldots,\1_{A_{i-1}},\1_{A\cup
A_i},\1_{A_{i+1}},\ldots,\1_{A_k})\\=\sum_{\substack{1\le j\le
k\\j\neq i}}b_j\|z_j\|_2^2+b_i\|z_i+w\|_2^2=\sum_{j=1}^n
b_j\|z_j\|_2^2+ 2b_i\langle
z_i,w\rangle+b_i\|w\|_2^2.\end{multline*} Thus $2\langle
z_i,w\rangle+\|w\|_2^2\le 0$, and if we sum this inequality over
$i\in \{1,\ldots,k\}$ while recalling that $w=-\sum_{i=1}^kz_i$ we
see that $(k-2)\|w\|_2^2\le 0$, which is a contradiction. Note that
for general $B$ the same argument shows that for all
$i\in\{1,\ldots,k\}$ we have $2\sum_{j=1}^k b_{ij} \left\langle
z_j,w\right\rangle+b_{ii}\|w\|_2^2\le 0$. These inequalities do not
seem to lend themselves to the same type of easy contradiction as in
the case of diagonal matrices. \fin}
\end{remark}

The proof of the following lemma is an obvious midification of the
proof of Lemma 3.2 in~\cite{KN08}.

\begin{lemma}\label{lem:old reduction to k-1}
If $n\ge k-1$ then $C(n,B)=C(k-1,B)$.
\end{lemma}

\begin{proof}
The inequality $C(n,B)\ge C(k-1,B)$ is easy since for every
$(f_1,\ldots,f_k)\in \Delta_k(\gamma_{k-1})$ we can define
$\left(\widetilde f_1,\ldots\widetilde f_k\right)\in
\Delta_k(\gamma_n)$ by $\widetilde f_j(x,y)=f_j(x)$ (thinking here
of $\R^n$ as $\R^{k-1}\times \R^{n-k+1}$). Then for all $j\in
\{1,\ldots,k\}$ we have
$\int_{\R^{k-1}}xf_j(x)d\gamma_{k-1}(x)=\int_{\R^n}x\widetilde
f_j(x)d\gamma_n(x)$, implying that $\Psi\left(\widetilde
f_1,\ldots\widetilde f_k\right)=\Psi\left(f_1,\ldots,f_k\right)$.

In the reverse direction, by Lemma~\ref{lem:maximizer exists} there
is a measurable partition $A_1,\ldots,A_k$ of $\R^n$ such that if we
define $z_j\coloneqq \int_{A_j}xd\gamma_n(x)\in \R^n$ then we have
$\sum_{i=1}^k\sum_{j=1}^k b_{ij}\left\langle
z_i,z_j\right\rangle=C(n,B)$. Note that $ \sum_{j=1}^k z_j=0$. Hence
the dimension of the subspace $V\coloneqq
\mathrm{span}\{z_1,\ldots,z_k\}$ is $d\le k-1$. Define
$g_1,\ldots,g_k:V\to [0,1]$ by $
g_j(x)=\gamma_{V^{\perp}}\left((A_j-x)\cap V^\perp\right)$. Then
$(g_1,\ldots,g_k)\in \Delta_k(\gamma_{V})$, so that
\begin{eqnarray*}
C(k-1,B)&\ge& C(d,B)\\&\ge& \sum_{i=1}^k\sum_{j=1}^k b_{ij}
\left\langle \int_{V} xg_i(x)d\gamma_V(x),\int_{V}
xg_j(x)d\gamma_V(x)\right\rangle\\&=&\sum_{i=1}^k\sum_{j=1}^k b_{ij}
\left\langle \int_{V}\int_{V^\perp}\1_{A_i}(x+y)
xd\gamma_V(x)d\gamma_{V^\perp}(y),\int_{V}\int_{V^\perp}\1_{A_j}(x+y)
xd\gamma_V(x)d\gamma_{V^\perp}(y)\right\rangle \\&=&
\sum_{i=1}^k\sum_{j=1}^k b_{ij} \left\langle
\int_{A_i}\mathrm{Proj}_{V}(w)d\gamma_n(w),\int_{A_j}\mathrm{Proj}_{V}(w)d\gamma_n(w)\right\rangle
\\&=&\sum_{i=1}^k\sum_{j=1}^k
b_{ij}\left\langle\mathrm{Proj}_{V}(z_i),\mathrm{Proj}_{V}(z_j)\right\rangle\\
&=& \sum_{i=1}^k\sum_{j=1}^k b_{ij}\left\langle
z_i,z_j\right\rangle=C(n,B),
\end{eqnarray*}
as required.
\end{proof}
In light of Lemma~\ref{lem:old reduction to k-1} we define
$C(B)\coloneqq C(k-1,B)$. We shall now prove an analogue of Lemma
3.3 in~\cite{KN08} which gives structural information on the
partition $\{A_1,\ldots, A_k\}$ of $\R^{k-1}$ at which $C(B)$ is
attained. We first recall some notation and terminology
from~\cite{KN08}. Given distinct $z_1,\ldots,z_k\in \R^{k-1}$ and
$j\in \{1,\ldots, k\}$ define a set $P_j(z_1,\ldots,z_k)\subseteq
\R^k$ by
$$
P_j(z_1,\ldots,z_k)\coloneqq \left\{x\in \R^{k}:\ \langle
x,z_j\rangle =\max_{i\in \{1,\ldots,k\}} \langle
x,z_i\rangle\right\}.
$$
Thus $\left\{P_j(z_1,\ldots,z_k)\right\}_{j=1}^k$ is a partition of
$\R^{k-1}$  which we call the simplicial partition induced by
$z_1,\ldots,z_k$ (strictly speaking the elements of this partition
are not disjoint, but they intersect at sets of measure $0$).

\begin{lemma}\label{lem:simplicial}  Let $A_1,\ldots,A_k\subseteq
\R^{k-1}$ be a partition into measurable sets such that  if we set
$z_j\coloneqq \int_{A_j}xd\gamma_{k-1}(x)$ then
\begin{equation}\label{eq:maximizer}C(B)=\sum_{i=1}^k\sum_{j=1}^k b_{ij}\left\langle
z_i,z_j\right\rangle.\end{equation}
 Assume also that this
partition is minimal in the sense that the number of elements of
positive measure in this partition is minimum among all the possible
partitions satisfying~\eqref{eq:maximizer}. Define $$J\coloneqq
\left\{j\in \{1,\ldots, k\}:\ \gamma_{k-1}(A_j)>0\right\}$$ and set
$|J|=\ell$. Then up to an orthogonal transformation $\{z_j\}_{j\in
J}\subseteq \R^{\ell-1}$ and the vectors $\{z_j\}_{j\in J}$ are
non-zero and distinct. Moreover, if we define $\{w_j\}_{j\in
J}\subseteq \R^{\ell-1}$ by
\begin{equation}\label{eq:def w_i}
w_j\coloneqq \sum_{s\in J} b_{js}z_s,
\end{equation}
then the vectors $\{w_j\}_{j\in J}$ are distinct and for each $j\in
J$ we have
\begin{equation}\label{eq:the cones}
A_j=P_j\big((w_i)_{i\in J}\big)\times \R^{k-\ell}
\end{equation}
up to sets of measure zero.
\end{lemma}

\begin{proof}
Since $\sum_{j\in J}\1_{A_j}=1$ almost everywhere we have
$\sum_{j\in J}z_j=0$. Thus the dimension of the span of
$\{z_j\}_{j\in J}$ is at most $|J|-1=\ell-1$, and by applying an
orthogonal transformation we may assume that $\{z_j\}_{j\in
J}\subseteq  \R^{\ell-1}$. Also, for every distinct $i,j\in J$
replace $A_i$ by $A_i\cup A_j$ and $A_j$ by the empty set and obtain
a partition of $\R^{k-1}$ which contains exactly $\ell-1$ elements
of positive measure and for which we have (by the minimality of
$\ell$):
\begin{eqnarray*}
C(B)&>& \sum_{s,t\in J\setminus\{i,j\}}b_{st}\left\langle
z_s,z_t\right\rangle+2\sum_{s\in J\setminus\{i,j\}}
b_{is}\left\langle z_s,z_i+z_j\right\rangle +
b_{ii}\left\|z_i+z_j\right\|_2^2\\ &=&
\sum_{s,t\in J}b_{st}\left\langle z_s,z_t\right\rangle+2\sum_{s\in
J}\left(b_{is}-b_{js}\right)\left\langle
z_s,z_j\right\rangle+\left(b_{ii}+b_{jj}-2b_{ij}\right)\|z_j\|_2^2\\
&=& C(B)+2\left\langle
w_i-w_j,z_j\right\rangle+\|v_i-v_j\|_2^2\cdot\|z_j\|_2^2,
\end{eqnarray*}
where we used the fact that $b_{st}=\left\langle
v_s,v_t\right\rangle$. Thus
\begin{equation}\label{eq:distinct ij}
2\left\langle
w_i-w_j,z_j\right\rangle+\|v_i-v_j\|_2^2\cdot\|z_j\|_2^2<0,
\end{equation}
and by symmetry we also have the inequality:
\begin{equation}\label{eq:distinct ij2}
2\left\langle
w_j-w_i,z_i\right\rangle+\|v_i-v_j\|_2^2\cdot\|z_i\|_2^2<0.
\end{equation}
It follows in particular from~\eqref{eq:distinct ij}
and~\eqref{eq:distinct ij2} that $z_i$ and $z_j$ are non-zero and
that $w_i\neq w_j$. Moreover if we sum~\eqref{eq:distinct ij}
and~\eqref{eq:distinct ij2} we get that
$$
2\left\langle
w_i-w_j,z_j-z_i\right\rangle+\|v_i-v_j\|_2^2\left(\|z_i\|_2^2+\|z_j\|_2^2\right)<0
$$
which implies that $z_i\neq z_j$.

The above reasoning implies in particular that
$\left\{P_j\big((w_i)_{i\in J}\big)\times \R^{k-\ell}\right\}_{j\in
J}$ is a partition of $\R^{k-1}$ (up to pairwise intersections at
sets of measure $0$). Assume for the sake of contradiction that
these exist
 $i\in J$ such that
$$
\gamma_{k-1}\left(A_i\setminus \left(P_i\big((w_s)_{s\in
J}\big)\times \R^{k-\ell}\right)\right)>0.
$$
Arguing as in the proof of Lemma 3.3 in~\cite{KN08} we see that
 there exists $\e>0$ and $j\in J\setminus \{i\}$
such that if we denote $E\coloneqq \left\{x\in A_i:\ \langle
x,z_j\rangle \ge \left\langle x,z_i\right\rangle+ \e\right\}$ then
$\gamma_{k-1}(E)>0$.

Define a partition $\widetilde A_1,\ldots \widetilde A_k$ of
$\R^{k-1}$ by
$$
\widetilde A_r\coloneqq \left\{\begin{array}{ll}A_r &
r\notin \{i,j\}\\
A_i\setminus E& r=i\\
A_j\cup E & r=j.\end{array}\right.
$$
Then for $w\coloneqq \int_{E}xd\gamma_{k-1}(x)$ we have
\begin{eqnarray*}
C(B)&\ge& \sum_{s,t\in J}b_{st} \left\langle\int_{\widetilde A_s}
xd\gamma_{k-1}(x),\int_{\widetilde A_t}
xd\gamma_{k-1}(x)\right\rangle\\&=&\sum_{s,t\in
J\setminus\{i,j\}}b_{st} \left\langle
z_s,z_t\right\rangle+2\sum_{s\in
J\setminus\{i,j\}}b_{is}\left\langle
z_s,z_i-w\right\rangle+2\sum_{s\in
J\setminus\{i,j\}}b_{js}\left\langle z_s,z_j+w\right\rangle
\\&\phantom{\le}&+2b_{ij}\left\langle z_i-w,z_j+w\right\rangle
+b_{ii}\|z_i-w\|_2^2+b_{jj}\|z_j+w\|_2^2
\\&=& C(B)-2\sum_{s\in J}b_{is}\left\langle
z_s,w\right\rangle+2\sum_{s\in J}b_{js}\left\langle
z_s,w\right\rangle+\left(b_{ii}+b_{jj}-2b_{ij}\right)\|w\|_2^2\\
&\stackrel{\eqref{eq:def w_i}}{=}& C(B)+2\left\langle
w_j-w_i,w\right\rangle+\|v_i-v_j\|_2^2\cdot\|w\|_2^2\\
&\ge& C(B)+2\int_{E} \left(\langle z_j,x\rangle-\langle
z_i,x\rangle\right)d\gamma_{k-1}(x)
\\&\ge&
C(B)+2\e\gamma_{k-1}(E)>C(B),
\end{eqnarray*}
a contradiction.
\end{proof}

\remove{
\begin{remark}\label{rem:if tilde} {\em In continuation of the discussion
in Remark~\ref{remark:simplex with 0} (and using the notation
there), if we maximize $\Psi$ on $\widetilde{\Delta_k(\gamma_n)}$
then the maximum is attained at $(\1_{A_1},\ldots,\1_{A_k})\in
\widetilde{\Delta_k(\gamma_n)}$ for some disjoint measurable sets
$A_1,\ldots,A_k\subseteq \R^n$ which do not necessarily cover
$\R^n$. Denote $B\coloneqq \R^n\setminus\bigcup_{i=1}^k A_i$. As
before we set $z_j\coloneqq \int_{A_j}xd\gamma_{k-1}(x)$ and let
$J\subseteq \{1,\ldots, k\}$ be set of indices $j$ for which $A_j$
has positive measure. If we assume once more that $|J|=\ell$ is
minimal then we obtain the following variant of
Lemma~\ref{lem:simplicial}: apply an orthogonal transformation so
that $\{z_j\}_{j\in J}\subseteq \R^\ell$ (note that we are no longer
ensured that the $z_j$ sum to $0$). Define $\{w_j\}_{j\in J}$ as
in~\eqref{eq:def w_i} and for $j\in J$ define
$$
D_j\coloneqq \left\{x\in R^\ell:\ \max_{i\in J} \langle
x,w_i\rangle=\langle x,w_j\rangle\ge 0\right\}.
$$
Then up to sets of measure $0$ we have $A_j=D_j\times \R^{n-\ell}$.
This is proved along the lines of the proof of
Lemma~\ref{lem:simplicial}, with the additional argument which
arises from the fact that we are allowed to change the sets by
transferring mass from some $A_j$ to $B$, and this would increase
the value of $\Psi$ if $\langle x,w_j\rangle< 0$.\fin}
\end{remark}
}

\begin{remark}\label{rem:crucial identities} {\em Note that we have the
following non-trivial identity as a corollary of
Lemma~\ref{lem:simplicial} (and using the same notation): For each
$i\in J$, \begin{equation}\label{eq:identity}
z_j=\int_{P_j\big((w_i)_{i\in J}\big)}xd\gamma_{\ell-1}(x),
\end{equation}
where we recall that the $w_i$ are defined in~\eqref{eq:def w_i}.
This system of equalities seems to contain non-trivial information
on the structure of the partition at which $C(B)$ is attained. In
future research it would be of interest to exploit this information,
though we have no need for it for our present purposes.\fin}
\end{remark}

\begin{remark}\label{rem:how to compute}
{\em Given $B$ and $\e>0$ we can estimate $C(B)$ up to an error of
at most $\e$ in constant time (which depends only on $B,k,\e$).
Moreover, we can compute in constant time a conical simplicial
partition of $\R^{k-1}$ at which the value of $\Psi$ is at least
$C(B)-\e$. These statements are a simple corollary of
Lemma~\ref{lem:simplicial}. Indeed, all we have to do is to run over
all choices of $\ell\in \{1,\ldots,k\}$ and for each such $\ell$
construct an appropriate net of $z_1,\ldots,z_\ell\in \R^{\ell-1}$
of bounded size, and then check each of the induced  simplicial
partitions of $\R^{k-1}$ as in~\eqref{eq:the cones} for the one
which maximizes $\Psi$. To this end we need some a priori bound on
the length of $z_i$: the crude bound
$$
\|z_i\|_2=\left\|\int_{A_i}xd\gamma_{\ell-1}(x)\right\|_2\le
\int_{\R^{\ell-1}} \|x\|_2d\gamma_{\ell-1}(x)\le \sqrt{\ell}
$$
will suffice. Fix $\delta>0$ which will be determined momentarily.
Let $\mathcal N$ be a $\delta$-net in the Euclidean ball of radius
$\sqrt{\ell}$ in $\R^{\ell-1}$. Then $|\mathcal{N}|\le
\left(\frac{3\sqrt{\ell}}{\delta}\right)^\ell$.}

{\em Let $A_1,\ldots,A_k$ be as in Lemma~\ref{lem:simplicial}, i.e.,
the true (minimal) partition at which $C(B)$ is attained. Let $J$,
$\ell$, $z_i$ and $w_i$ be as in Lemma~\ref{lem:simplicial}. For
each $i\in J$ find $z_i'\in \mathcal N$ for which $\|z_i-z_i'\|_2\le
\delta$. Define $w_i'=\sum_{s\in J} b_{js}z_s'$. Then we have the
crude bound $\|w_i-w_i'\|_2\le \delta\sum_{s=1}^k\sum_{t=1}^k
|b_{st}|\coloneqq\delta \|B\|_1$. We also have the a priori bounds
$\|w_i\|_2,\|w_i'\|_2\le \sqrt{\ell}\|B\|_1$. By compactness there
exists $\delta=\delta(\e,\ell,B)$ such that these estimates imply
that for all $j\in J$,
\begin{equation}\label{eq:compactness}
\left\|z_j-\int_{P_j\big((w_i')_{i\in
J}\big)}xd\gamma_{\ell-1}(x)\right\|_2=\left\|\int_{P_j\big((w_i)_{i\in
J}\big)}xd\gamma_{\ell-1}(x)-\int_{P_j\big((w_i')_{i\in
J}\big)}xd\gamma_{\ell-1}(x)\right\|_2 \le
\frac{\e}{2\sqrt{\ell}\|B\|_1}.
\end{equation}
(It is actually easy to give a concrete bound on the required
$\delta$ if so desired, but this is not important for our purposes.)
It follows from~\eqref{eq:compactness} that:
\begin{multline*}
C(B)\ge \sum_{s,t\in J}b_{st} \left\langle\int_{P_s\big((w_i')_{i\in
J}\big)}xd\gamma_{\ell-1}(x),\int_{P_t\big((w_i')_{i\in
J}\big)}xd\gamma_{\ell-1}(x)\right\rangle\\
\ge \sum_{s,t\in J}b_{st} \left\langle
z_s,z_t\right\rangle-\sum_{s,t\in J} |b_{st}| \cdot
\frac{\e}{2\sqrt{\ell}\|B\|_1}\cdot 2\sqrt{\ell}=C(B)-\e.
\end{multline*}
Note that the above integrals can be estimated efficiently
(polynomial time in $k$) with arbitrarily good precision due to the
fact that the simplicial cones $P_j\big((w_i')_{i\in J}\big)$ have
an efficient membership oracle and the Gaussian measure is
$\log$-concave. These are very crude bounds that suffice for our
algorithmic purposes when $k$ is fixed, but deteriorate
exponentially with $k$. It would be of interest to understand
whether we can estimate $C(B)$ (and more importantly the associated
partitions, as they are used in our rounding procedure) in time
which is polynomial in $k$. Perhaps the
identities~\eqref{eq:identity} can play a role in the design of such
an efficient algorithm, but we did not investigate this issue. \fin}
\end{remark}

We end this section with a simple analytic interpretation of the
parameter $C(B)$. Given a square integrable function $f: \R^n\to
\R^k$ its Rademacher projection $\Rad (f):\R^n\to \R^k$
(see~\cite{Pisier89} for an explanation of this terminology) is
defined for $x=(x_1,\ldots,x_n)\in \R^n$ as:
$$
\Rad (f)(x)=\sum_{i=1}^n
\left(\int_{\R^n}y_if(y)d\gamma_n(y)\right)x_i.
$$
Assume that $f$ takes values in $\{v_1,\ldots,v_k\}\subseteq \R^k$
and define $A_i=f^{-1}(v_i)$ for $i\in \{1,\ldots,k\}$. Then
$\{A_1,\ldots,A_k\}$ is a measurable partition of $\R^n$. We also
have the identity:
$$
\Rad (f)(x)=\sum_{i=1}^n
\left(\sum_{j=1}^kv_j\int_{A_j}y_id\gamma_n(y)\right)x_i.
$$
Thus
\begin{multline}\label{eq:rademacher identity}
\left\|\Rad
(f)\right\|_{L_2(\gamma_n,\R^k)}^2=\int_{\R^n}\left\|\Rad
(f)(x)\right\|_2^2d\gamma_n(x)=\sum_{i=1}^n\left\|\sum_{j=1}^kv_j\int_{A_j}y_id\gamma_n(y)\right\|_2^2\\
=\sum_{i=1}^n\sum_{j=1}^k\sum_{\ell=1}^k \left\langle
v_j,v_\ell\right\rangle\left(\int_{A_j}y_id\gamma_n(y)\right)\left(\int_{A_\ell}y_id\gamma_n(y)\right)
=\sum_{j=1}^k\sum_{\ell=1}^kb_{j\ell}\left\langle\int_{A_j}yd\gamma_n(y),\int_{A_j}yd\gamma_n(y)\right\rangle.
\end{multline}
The identity~\eqref{eq:rademacher identity} implies the following
lemma:
\begin{lemma}\label{lem:rademacher} For every $n\ge k-1$ we have:
$$ C(B)=\max_{f:\R^n\to \{v_1,\ldots,v_k\}} \left\|\Rad
(f)\right\|_{L_2(\gamma_n,\R^k)}^2.
$$
\end{lemma}
Recall that $R(B)$ is defined as the radius of the smallest ball in
$\R^k$ which contains the set $\{v_1,\ldots,v_k\}$ and that $w(B)$
is the center of this ball. Lemma~\ref{lem:rademacher} implies the
following corollary:

\begin{corollary}\label{cor:CR^2}
$C(B)\le R(B)^2$.
\end{corollary}

\begin{proof} Let $\{A_1,\ldots,A_k\}$ be a partition of $\R^{k-1}$
into measurable sets such that if we define
$z_j=\int_{A_j}xd\gamma_{k-1}(x)$ then
\begin{multline}\label{eq:before projection}
C(B)=\sum_{i=1}^k\sum_{j=1}^k \left\langle v_i,v_j \right\rangle
\left\langle z_i,z_j \right\rangle\\=\sum_{i=1}^k\sum_{j=1}^k
\left\langle v_i-w(B),v_j-w(B) \right\rangle \left\langle z_i,z_j
\right\rangle+2\sum_{i=1}^k \left\langle v_i,w(B)
\right\rangle\left\langle z_i,\sum_{j=1}^k
z_j\right\rangle+\|w(B)\|_2^2\cdot\left\|\sum_{j=1}^k
z_j\right\|_2^2.
\end{multline}
Since $\sum_{j=1}^kz_j=0$ it follows from~\eqref{eq:rademacher
identity} and~\eqref{eq:before projection} that for
$f:\R^{k-1}\to\left\{v_i-w(B)\right\}_{i=1}^k$ defined by
$f|_{A_i}=v_i-w(B)$ we have:
$$
C(B)=\left\|\Rad
(f)\right\|_{L_2(\gamma_n,\R^k)}^2\stackrel{(\star)}{\le}
\left\|f\right\|_{L_2(\gamma_n,\R^k)}^2\le
\left\|f\right\|_{L_\infty(\gamma_n,\R^k)}^2=\max_{i\in
\{1,\ldots,k\}}\left\|v_i-w(B)\right\|_2^2=R(B)^2,
$$
where in $(\star)$ we used the fact that $\Rad$ is an orthogonal
projection on the Hilbert space $L_2(\gamma_n,\R^k)$.
\end{proof}

\section{Generalized positive semidefinite Grothendieck
inequalities}\label{sec:ineq}

The purpose of this section is to prove the following theorem, which
as explained in the introduction, is an extension of Grothendieck's
inequality for positive semidefinite matrices.

\begin{theorem}\label{them:our grothendieck}
Let $A=(a_{ij})\in M_n(\R)$ be an $n\times n$ symmetric positive
semidefinite matrix. Let $v_1,\ldots,v_k\in \R^k$ be $k\ge 2$
vectors and let $B=(b_{ij}=\langle v_i,v_j\rangle)$ be the
corresponding Gram matrix. Then
\begin{equation}\label{eq:our gro}
\max_{x_1,\ldots,x_n\in S^{n-1}}\sum_{i=1}^n\sum_{j=1}^n a_{ij}
\langle x_i,x_j\rangle \le
\frac{1}{C(B)}\max_{\sigma:\{1,\ldots,n\}\to
\{1,\ldots,k\}}\sum_{i=1}^n\sum_{j=1}^n a_{ij}\langle
v_{\sigma(i)},v_{\sigma(j)}\rangle.
\end{equation}
\end{theorem}
We shall prove in Section~\ref{sec:optimal} that the factor
$\frac{1}{C(B)}$ in~\eqref{eq:our gro} cannot be improved, even when
in~\eqref{eq:our gro} $A$ is restricted to be centered, i.e.,
$\sum_{i=1}^n\sum_{j=1}^n a_{ij}=0$.

The key tool in the proof of Theorem~\ref{them:our grothendieck} is
the following lemma:

\begin{lemma}\label{lem:hermite} Let $\left\{g_{ij}:\
i\in\{1,\ldots,m\},\ j\in \{1,\ldots,n\}\right\}$ be i.i.d. standard
Gaussian random variables and let $G=(g_{ij})$ be the corresponding
$m\times n$ random Gaussian matrix. Fix two unit vectors $x,y\in
S^{n-1}$ and two measurable subsets $E,F\subseteq \R^m$. Then:
\begin{multline}\label{poisson identity}
\Pr\left[Gx\in E\ \wedge\ Gy\in
F\right]\\=\gamma_m(E)\gamma_m(F)+\langle
x,y\rangle\left\langle\int_Eud\gamma_m(u),\int_Fud\gamma_m(u)\right\rangle+\sum_{\ell=2}^\infty
\left\langle x^{\otimes \ell},y^{\otimes
\ell}\right\rangle\sum_{\substack{s\in
(\N\cup\{0\})^m\\s_1+\cdots+s_m=\ell}}\alpha_s(E)\alpha_s(F),
\end{multline}
for some real coefficients $\{\alpha_s(E)\}_{s\in
(\N\cup\{0\})^m},\{\alpha_s(F)\}_{s\in (\N\cup\{0\})^m}\subseteq
\R$.
\end{lemma}
\begin{proof}
Denote $r=\langle x,y\rangle$. Let $g,h\in \R$ be independent
standard Gaussian random variables and let $g_1,\ldots,g_m\in \R^n$
 be i.i.d. standard Gaussian random vectors in $\R^n$ (i.e., they are independent and distributed according to
 $\gamma_n$). Then for each $i\in \{1,\ldots,m\}$ the planar random vector
 $(\langle g_i,x\rangle,\langle g_i,y\rangle)\in \R^2$ has the same
 distribution as $\left(g,rg+\sqrt{1-r^2}h\right)\in \R^2$, and
 hence its density is given for $(u,v)\in \R^2$ by:
$$
f_r(u,v)\coloneqq
\frac{1}{2\pi\sqrt{1-r^2}}\cdot\exp\left(-\frac{u^2-2ruv+v^2}{2(1-r^2)}\right).
$$
The Hermite polynomials $\{H_k\}_{k=0}^\infty$ are defined as:
$$
H_k(t)\coloneqq
(-1)^ke^{t^2}\frac{d^k}{dt^k}\left(e^{-t^2}\right)=\sum_{s=0}^{\lfloor
k/2\rfloor}\frac{(-1)^sk!}{s!(k-2s)!}(2t)^{k-2s}.
$$
The formula for the Poison kernel for Hermite polynomials (see for
example equation 6.1.13 in~\cite{AAR99} or the discussion
in~\cite{Stein93}) says that
$$
f_r(u,v)=\frac{e^{-(u^2+v^2)/2}}{2\pi}\sum_{k=0}^\infty
\frac{r^k}{2^kk!}H_k\left(\frac{u}{\sqrt{2}}\right)H_k\left(\frac{v}{\sqrt{2}}\right).
$$
Since the vector $(Gx,Gy)\in \R^{2m}$ has the same distribution as
the vector $\big((\langle g_i,x\rangle,\langle
g_i,y\rangle)\big)_{i=1}^m$, whose (planar) entries are i.i.d. with
density $f_r$, we see that:

\begin{eqnarray*}
&&\!\!\!\!\!\!\!\!\!\!\!\!\Pr\left[Gx\in E\ \wedge\ Gy\in
F\right]=\int_{E\times F} \left(\prod_{i=1}^m
f_r(u_i,v_i)\right)dudv\\&=&\int_{E\times F}
\frac{e^{-(\|u\|_2^2+\|v\|_2^2)/2}}{(2\pi)^m}\left(\prod_{i=1}^m\left(
\sum_{k=0}^\infty
\frac{r^k}{2^kk!}H_k\left(\frac{u_i}{\sqrt{2}}\right)H_k\left(\frac{v_i}{\sqrt{2}}\right)\right)\right)dudv\\&=&
\int_{E\times F}\left(\sum_{s\in
(\N\cup\{0\})^m}\frac{r^{s_1+\cdots+s_m}}{2^{s_1+\cdots+s_m}\prod_{i=1}^m
s_i!}\left(\prod_{i=1}^mH_{s_i}
\left(\frac{u_i}{\sqrt{2}}\right)\right)\left(\prod_{i=1}^mH_{s_i}\left(\frac{v_i}{\sqrt{2}}\right)\right)\right)d\gamma_m(u)d\gamma_m(v)\\
&=& \gamma_m(E)\gamma_m(F)+\langle
x,y\rangle\left\langle\int_Eud\gamma_m(u),\int_Fud\gamma_m(u)\right\rangle+\sum_{\ell=2}^\infty
\left\langle x^{\otimes \ell},y^{\otimes
\ell}\right\rangle\sum_{\substack{s\in
(\N\cup\{0\})^m\\s_1+\cdots+s_m=\ell}}\alpha_s(E)\alpha_s(F),
\end{eqnarray*}
where we used the fact that $H_0(t)=1$ and $H_1(t)=2t$, and for
every measurable subset $W\subseteq R^m$ and $s\in (\N\cup\{0\})^m$
the notation
$$
\alpha_s(W)\coloneqq \frac{1}{2^{(s_1+\cdots+s_m)/2}\prod_{i=1}^m
\sqrt{s_i!}}\int_W \left(\prod_{i=1}^mH_{s_i}
\left(\frac{u_i}{\sqrt{2}}\right)\right)d\gamma_m(u).
$$
The proof of the identity~\eqref{eq:identity} is complete.
\end{proof}

\begin{proof}[Proof of Theorem~\ref{them:our grothendieck}]
Fix $n$ unit vectors $x_1,\ldots,x_n\in S^{n-1}$. Let
$\{A_1,\ldots,A_k\}$ be a partition of $\R^{k-1}$ into measurable
subsets. Let $G$ be a random Gaussian matrix as in
Lemma~\ref{lem:hermite} with $m=k-1$. Define a random assignment
$\sigma:\{1,\ldots,n\}\to\{1,\ldots,k\}$ by setting $\sigma(i)$ to
be the unique $p\in \{1,\ldots,k\}$ for which $Gx_i\in A_p$. Then
for every $i,j\in \{1,\ldots,n\}$ we have
$$
\E \left[\left\langle
v_{\sigma(i)},v_{\sigma(j)}\right\rangle\right]=\sum_{p=1}^k\sum_{q=1}^k
\left\langle v_{p},v_{q}\right\rangle\Pr\left[Gx_i\in A_p\ \wedge\
Gx_j\in A_q\right]=\sum_{p=1}^k\sum_{q=1}^k b_{pq}\Pr\left[Gx_i\in
A_p\ \wedge\ Gx_j\in A_q\right].
$$
We may therefore apply Lemma~\ref{lem:hermite} to deduce that:
\begin{eqnarray*}
\E\left[\sum_{i=1}^n\sum_{j=1}^na_{ij}\left\langle
v_{\sigma(i)},v_{\sigma(j)}\right\rangle\right]&=&\left(\sum_{i=1}^n\sum_{j=1}^na_{ij}\right)
\sum_{p=1}^k\sum_{q=1}^k
b_{pq}\gamma_{k-1}(A_p)\gamma_{k-1}(A_q)\\&\phantom{\le}&+
\left(\sum_{i=1}^n\sum_{j=1}^na_{ij}\left\langle
x_i,x_j\right\rangle\right)\sum_{p=1}^k
\sum_{q=1}^kb_{pq}\left\langle\int_{A_p}xd\gamma_{k-1}(x),\int_{A_q}xd\gamma_{k-1}(x)\right\rangle\\&\phantom{\le}&+
\sum_{\ell=2}^\infty
\left(\sum_{i=1}^n\sum_{j=1}^na_{ij}\left\langle x_i^{\otimes
\ell},x_j^{\otimes \ell}\right\rangle\right)\sum_{\substack{s\in
(\N\cup\{0\})^m\\s_1+\cdots+s_m=\ell}}\sum_{p=1}^k
\sum_{q=1}^kb_{pq}\alpha_s(A_p)\alpha_s(A_q)\\
&\ge& \left(\sum_{i=1}^n\sum_{j=1}^na_{ij}\left\langle
x_i,x_j\right\rangle\right)\sum_{p=1}^k
\sum_{q=1}^kb_{pq}\left\langle\int_{A_p}xd\gamma_{k-1}(x),\int_{A_q}xd\gamma_{k-1}(x)\right\rangle,
\end{eqnarray*}
where we used the fact that both $A$ and $B$ are positive
semidefinite. It thus follows that there exists an assignment
$\sigma:\{1,\ldots,n\}\to\{1,\ldots,k\}$ for which
$$
\sum_{i=1}^n\sum_{j=1}^na_{ij}\left\langle
v_{\sigma(i)},v_{\sigma(j)}\right\rangle\ge
\left(\sum_{i=1}^n\sum_{j=1}^na_{ij}\left\langle
x_i,x_j\right\rangle\right)\sum_{p=1}^k
\sum_{q=1}^kb_{pq}\left\langle\int_{A_p}xd\gamma_{k-1}(x),\int_{A_q}xd\gamma_{k-1}(x)\right\rangle,
$$
and since this is true for all measurable partitions
$\{A_1,\ldots,A_k\}$ of $\R^{k-1}$ we deduce that there exists an
assignment $\sigma:\{1,\ldots,n\}\to\{1,\ldots,k\}$ for which:
$$
\sum_{i=1}^n\sum_{j=1}^na_{ij}\left\langle
v_{\sigma(i)},v_{\sigma(j)}\right\rangle\ge
C(B)\sum_{i=1}^n\sum_{j=1}^na_{ij}\left\langle x_i,x_j\right\rangle,
$$
as required.
\end{proof}

\subsection{Optimality}\label{sec:optimal}

The purpose of this section is to show that Theorem~\ref{them:our
grothendieck} is sharp:

\begin{theorem}\label{thm:sharp}
 Let $v_1,\ldots,v_k\in \R^k$ be $k\ge 2$
vectors and let $B=(b_{ij}=\langle v_i,v_j\rangle)$ be the
corresponding Gram matrix. Assume that $K>0$ is a constant such that
for every $n\in \N$ and every centered symmetric positive
semidefinite matrix $A=(a_{ij})\in M_n(\R)$ we have:
\begin{equation}\label{eq:our gro with K}
\max_{x_1,\ldots,x_n\in S^{n-1}}\sum_{i=1}^n\sum_{j=1}^n a_{ij}
\langle x_i,x_j\rangle \le K\max_{\sigma:\{1,\ldots,n\}\to
\{1,\ldots,k\}}\sum_{i=1}^n\sum_{j=1}^n a_{ij}\langle
v_{\sigma(i)},v_{\sigma(j)}\rangle.
\end{equation}
Then $K\ge \frac{1}{C(B)}$.
\end{theorem}

\begin{proof}
The proof consists of a discretization of a continuous example. The
discretization step is somewhat tedious, but straightforward. We
will start with a presentation of the continuous example. Fix $m\in
\N$ and let $g,h\in \R^m$ be independent standard gaussian random
vectors. Since $(\|g\|_2,\|h\|_2)$ is independent of
$\left(\frac{g}{\|g\|_2},\frac{h}{\|h\|_2}\right)$ we have:
\begin{multline}\label{eq;cont kernel}
\int_{\R^m\times \R^m} \langle x,y\rangle\cdot
\left\langle\frac{x}{\|x\|_2},\frac{y}{\|y\|_2}\right\rangle
d\gamma_m(x)d\gamma_m(y)=\E\left[\|g\|_2\cdot\|h\|_2\left\langle\frac{g}{\|g\|_2},
\frac{h}{\|h\|_2}\right\rangle^2\right]\\=
\E\left[\|g\|_2\cdot\|h\|_2\right]\cdot
\E\left[\left\langle\frac{g}{\|g\|_2},
\frac{h}{\|h\|_2}\right\rangle^2\right]=\E\left[\|g\|_2\right]^2\E\left[\frac{g_1^2}{\|g\|_2^2}\right]=
\E\left[\|g\|_2\right]^2\frac{1}{m}\sum_{i=1}^m\E\left[\frac{g_i^2}{\|g\|_2^2}\right]=
\frac{1}{m}\E\left[\|g\|_2\right]^2,
\end{multline}
where we used the rotation invariance of the distribution of $h$.

The distribution of $\|g\|_2^2$ is the $\chi^2$ distribution with
$m$ degrees of freedom, and therefore its density at $u>0$ equals $
\frac{1}{2^{m/2}\Gamma(m/2)}u^{\frac{m}{2}-1}e^{-u/2}$. It follows
that
\begin{equation}\label{eq:stirling}
\E\left[\|g\|_2\right]=\frac{1}{2^{m/2}\Gamma(m/2)}\int_0^\infty
\sqrt{u}\cdot u^{\frac{m}{2}-1}e^{-u/2} du=\sqrt{2}\cdot
\frac{\Gamma\left(\frac{m+1}{2}\right)}{\Gamma\left(\frac{m}{2}\right)}\ge
\sqrt{m}\left(1-O\left(\frac{1}{m}\right)\right),
\end{equation}
where the last step is an application of Stirling's formula.
Plugging~\eqref{eq:stirling} into~\eqref{eq;cont kernel} we see
that:
\begin{equation}\label{eq:conclusion cont kernel}
 \int_{\R^m\times \R^m} \langle x,y\rangle\cdot
\left\langle\frac{x}{\|x\|_2},\frac{y}{\|y\|_2}\right\rangle
d\gamma_m(x)d\gamma_m(y)\ge 1-O\left(\frac{1}{m}\right).
\end{equation}

Now, assuming that $m\ge k-1$, for every $f:\R^m\to
\{v_1,\ldots,v_k\}$ we have
\begin{multline}\label{eq:continuous RHS}
\int_{\R^m\times \R^m} \langle x,y\rangle\cdot \left\langle
f(x),f(y)\right\rangle d\gamma_m(x)d\gamma_m(y)=\left\|\int_{\R^m}
x\otimes f(x)d\gamma_m(x)\right\|_2^2\\=\left\|\sum_{i=1}^m
e_i\otimes \left(\int_{\R^m} x_i
f(x)d\gamma_m(x)\right)\right\|_2^2=\left\|\Rad(f)\right\|_{L_2(\gamma_m,\R^k)}^2\le
C(B),
\end{multline}
where we used Lemma~\ref{lem:rademacher} (and here $e_1,\ldots,e_m$
is the standard basis or $\R^m$).

We shall now perform a simple discretization argument to conclude
the proof of Theorem~\ref{thm:sharp}. Fix $\e>0$ and $M\in \N$. Let
$\F$ be the set of all axis parallel cubes in $[-\e M,\e M]^m$ which
are a product of $m$ intervals whose endpoints are consecutive
integer multiples of $\e$ in $[-M,M]$. Thus $|\F|=(2M)^m$ and each
$Q\in \F$ has volume $\e^m$. For  $Q\in \F$ let $z_Q$ be the center
of $Q$. For every $P,Q\in \F$ define
$$
a_{PQ}\coloneqq
\e^{2m}e^{-\frac{\|z_P\|_2^2+\|z_Q\|_2^2}{2}}\left\langle
z_P,z_Q\right\rangle.
$$
By our assumption~\eqref{eq:our gro with K} there is an assignment
$\sigma:\F\to \{1,\ldots,k\}$ such that
\begin{equation}\label{eq:use assumtion sharpness}
\sum_{P,Q\in \F} a_{PQ} \left\langle
\frac{z_P}{\|z_P\|_2},\frac{z_Q}{\|z_Q\|_2}\right\rangle\le
K\sum_{P,Q\in \F} a_{PQ}\left\langle
v_{\sigma(P)},v_{\sigma(Q)}\right\rangle.
\end{equation}
We shall now use the following straightforward (and crude)
estimates:
\begin{eqnarray*}
&&\!\!\!\!\!\!\!\!\!\!\!\!\!\!\!\!\left|\int_{\R^m\times \R^m}
\langle x,y\rangle
\left\langle\frac{x}{\|x\|_2},\frac{y}{\|y\|_2}\right\rangle
d\gamma_m(x)d\gamma_m(y)-\sum_{P,Q\in \F} a_{PQ} \left\langle
\frac{z_P}{\|z_P\|_2},\frac{z_Q}{\|z_Q\|_2}\right\rangle\right|\\
&\le& \sum_{P,Q\in \F} \int_{P\times
Q}\left|e^{-\frac{\|z_P\|_2^2+\|z_Q\|_2^2}{2}}\left\langle
z_P,z_Q\right\rangle\left\langle
\frac{z_P}{\|z_P\|_2},\frac{z_Q}{\|z_Q\|_2}\right\rangle-
e^{-\frac{\|x\|_2^2+\|y\|_2^2}{2}}\left\langle
x,y\right\rangle\left\langle
\frac{x}{\|x\|_2},\frac{y}{\|y\|_2}\right\rangle\right|dxdy\\
&\phantom{\le}&+\left|\int_{(\R^m\times \R^m)\setminus ([-\e M,\e
M]^m\times [-\e M,\e M]^m)} \langle x,y\rangle
\left\langle\frac{x}{\|x\|_2},\frac{y}{\|y\|_2}\right\rangle
d\gamma_m(x)d\gamma_m(y)\right|\\
&\le& O(1) \sqrt{m}\e\left(\sqrt{m}M\e\right)^{3}\sum_{P,Q\in \F}
\int_{P\times Q}e^{-\frac{\|x\|_2^2+\|y\|_2^2}{2}}dx
dy+O(1)m^2e^{-\frac{(\e M)^2}{4}}\\
&\le& O(1)
\sqrt{m}\e\left(\sqrt{m}M\e\right)^{3}+O(1)m^2e^{-\frac{(\e
M)^2}{4}}.
\end{eqnarray*}
We shall require in what follows that $\e M=2m$. Hence,
using~\eqref{eq:conclusion cont kernel} we deduce that:
\begin{equation}\label{eq:conclusion discretization cont kernel}
\sum_{P,Q\in \F} a_{PQ} \left\langle
\frac{z_P}{\|z_P\|_2},\frac{z_Q}{\|z_Q\|_2}\right\rangle\ge
1-O\left(m^5\e+\frac{1}{m}\right).
\end{equation}
On the other hand, define $f:\R^m\to \{v_1,\ldots,v_k\}$ by
$$
f(x)=\left\{\begin{array}{ll}
 v_{\sigma(Q)}&  x\in Q\in\F,\\
 v_1 & x\notin [-\e M,\e M]^m.\end{array}\right.
$$
Observe that by symmetry
$$
\int_{(\R^m\times \R^m)\setminus \left([-\e M,\e M]^m\times [-\e
M,\e M]^m\right)}\langle x,y\rangle\cdot \left\langle
f(x),f(y)\right\rangle d\gamma_m(x)d\gamma_m(y)=0,
$$
and therefore a similar crude estimate yields:
\begin{eqnarray}\label{eq;discretize RHS}
&&\!\!\!\!\!\!\!\!\!\!\!\!\!\!\!\!\nonumber\left|\int_{\R^m\times
\R^m} \langle x,y\rangle\cdot \left\langle f(x),f(y)\right\rangle
d\gamma_m(x)d\gamma_m(y)-\sum_{P,Q\in \F} a_{PQ}\left\langle
v_{\sigma(P)},v_{\sigma(Q)}\right\rangle\right|\\
&\le&\nonumber\sum_{P,Q\in \F} \int_{P\times Q}
\left|e^{-\frac{\|x\|_2^2+\|y\|_2^2}{2}}\left\langle
x,y\right\rangle-e^{-\frac{\|z_P\|_2^2+\|z_Q\|_2^2}{2}}\left\langle
z_P,z_Q\right\rangle\right|\left|\left\langle
v_{\sigma(P)},v_{\sigma(Q)}\right\rangle\right|dxdy\\&\le& O\left(
m^5\e\right)\max_{i\in \{1,\ldots,k\}}\|v_i\|_2^2.
\end{eqnarray}
Choosing $\e=m^{-6}$ (and thus $M=2m^7$), and
combining~\eqref{eq;discretize RHS} with~\eqref{eq:continuous RHS}
and~\eqref{eq:conclusion discretization cont kernel}, yields in
combination with~\eqref{eq:use assumtion sharpness} the bound:
$$
1-O\left(\frac{1}{m}\right)\le
K\left(C(B)+O\left(\frac{1}{m}\right)\max_{i\in
\{1,\ldots,k\}}\|v_i\|_2^2\right).
$$
Letting $m\to \infty$ concludes the proof of
Theorem~\ref{thm:sharp}.
\end{proof}

\remove{

\begin{corollary}\label{C(B) lower} If $B=(b_{ij})\in
M_k(\R)$ is symmetric, positive semidefinite, centered
($\sum_{i=1}^k\sum_{j=1}^k b_{ij}=0$) and spherical ($b_{ii}=1$ for
all $i\in \{1,\ldots,k\}$) then:
$$C(B)\ge\frac{9k}{8\pi(k-1)}.$$
\end{corollary}

\begin{proof}
In~\cite{KN08} it was shown that under these
assumptions~\eqref{eq:our gro with K} holds with
$K=\frac{8\pi}{9}\left(1-\frac{1}{k}\right)$, so that the required
result follows from Theorem~\ref{thm:sharp}.
\end{proof}}

\section{A sharp approximation algorithm for kernel
clustering}\label{sec:alg}

Let $A=(a_{ij})\in M_n(\R)$ be a centered symmetric positive
semidefinite matrix and let $B=(b_{ij})\in M_k(\R)$ be a symmetric
positive semidefinite matrix. Our goal is to design a polynomial
time algorithm which approximates the value:
$$ \Clust(A|B) = \max_{\sigma: \{1,\ldots,n\} \to \{1,\ldots,k\}}
  \sum_{i=1}^n\sum_{j=1}^n a_{ij} b_{\sigma(i) \sigma(j)}.
  $$
We proceed as follows. We first find vectors $v_1,\ldots,v_k\in
\R^k$ such that $b_{ij}=\langle v_i,v_j\rangle$ for all $i,j\in
\{1,\ldots, k\}$. This can be done in polynomial time (Cholesky
decomposition). Let $R(B)$ be the minimum radius of the Euclidean
ball in $\R^k$ that contains $\{v_1,\ldots,v_k\}$ and let $w(B)$ be
the center of this ball. Both $R(B)$ and $w(B)$ can be efficiently
computed by solving an appropriate semidefinite program.

We now use semidefinite programming to compute the value:
\begin{multline}\label{def:SDP}
\SDP(A|B)\coloneqq \max\left\{\sum_{i=1}^n\sum_{j=1}^n
a_{ij}\left\langle x_i,x_j\right \rangle :\ x_1,\ldots,x_n\in \R^n\
\wedge\ \|x_i\|_2\le 1\ \forall
i\in\{1,\ldots,n\}\right\}\\=\max\left\{\sum_{i=1}^n\sum_{j=1}^n
a_{ij}\left\langle x_i,x_j\right \rangle :\ x_1,\ldots,x_n\in
S^{n-1}\right\},
\end{multline}
where the last equality in~\eqref{def:SDP} holds since the function
$(x_1,\ldots,x_n)\mapsto \sum_{i=1}^n\sum_{j=1}^n a_{ij}\left\langle
x_i,x_j\right \rangle$ is convex (by virtue of the fact that $A$ is
positive semidefinite). We claim that
\begin{equation}\label{eq:guarantee}
\frac{\Clust(A|B)}{R(B)^2}\le \SDP(A|B)\le \frac{\Clust(A|B)}{C(B)},
\end{equation}
which implies that if we output the number $R(B)^2\Clust(A|B)$ we
will obtain a polynomial time algorithm which approximates
$\Clust(A|B)$ up to a factor of $\frac{R(B)^2}{C(B)}$.

To verify~\eqref{eq:guarantee} let $x_1^*,\ldots,x_n^*\in S^{n-1}$
and $\sigma^*:\{1,\ldots,n\}\to \{1,\ldots,k\}$ be such that
$$
\SDP(A|B)=\sum_{i=1}^n\sum_{j=1}^n a_{ij}\left\langle
x_i^*,x_j^*\right \rangle,
$$
and
$$
\Clust(A|B)=\sum_{i=1}^n\sum_{j=1}^n a_{ij} b_{\sigma^*(i)
\sigma^*(j)}.
$$

Write $(a_{ij})_{i,j=1}^n=(\langle u_i,u_j\rangle)_{i,j=1}^n$ for
some $u_1,\ldots,u_n\in\R^n$. The assumption that $A$ is centered
means that $\sum_{i=1}^nu_i=0$. The right-hand side of inequality
in~\eqref{eq:guarantee} is simply a restatement of
Theorem~\ref{them:our grothendieck}. The left-hand side
inequality~\eqref{eq:guarantee} follows from the fact that
$\frac{v_{\sigma^*(i)}-w(B)}{R(B)}$ has norm at most $1$ for all
$i\in \{1,\ldots,n\}$. Indeed, these norm bounds imply that:
\begin{eqnarray*}
\SDP(A|B)&\ge&  \sum_{i=1}^n\sum_{j=1}^n a_{ij}\left\langle
\frac{v_{\sigma^*(i)}-w(B)}{R(B)},\frac{v_{\sigma^*(j)}-w(B)}{R(B)}\right
\rangle\\
&=& \frac{1}{R(B)^2} \sum_{i=1}^n\sum_{j=1}^n a_{ij}\left\langle
v_{\sigma^*(i)},v_{\sigma^*(j)}\right
\rangle-\frac{2}{R(B)^2}\sum_{i=1}^n\left\langle
w(B),v_{\sigma^*(i)}\right\rangle\left\langle u_i,\sum_{j=1}^nu_j
\right\rangle+\frac{\|w(B)\|_2^2}{R(B)^2} \sum_{i=1}^n\sum_{j=1}^n
a_{ij}\\
&=& \frac{\Clust(A|B)}{R(B)^2}.
\end{eqnarray*}
\medskip

This completes the proof that our algorithm approximates efficiently
the number $\Clust(A|B)$, but does not address the issue of how to
efficiently compute an assignment $\sigma:\{1,\ldots,n\}\to
\{1,\ldots,k\}$ for which the induced clustering of $A$ has the
required value. An inspection of the proof of Theorem~\ref{them:our
grothendieck} shows that the issue here is to find efficiently a
conical simplicial partition $A_1,\ldots,A_k$ of $\R^{k-1}$ at which
$C(B)$ is almost attained, say
$$
\sum_{p=1}^k
\sum_{q=1}^kb_{pq}\left\langle\int_{A_p}xd\gamma_{k-1}(x),\int_{A_q}xd\gamma_{k-1}(x)\right\rangle\ge
(1-\e)C(B).
$$
Once this partition is computed, using the notation in the proof of
Theorem~\ref{them:our grothendieck} we have a randomized algorithm
which outputs an assignment $\sigma:\{1,\ldots,n\}\to
\{1,\ldots,k\}$ such that
$$
\E_\sigma\left[\sum_{i=1}^n\sum_{j=1}^n
a_{ij}b_{\sigma(i)\sigma(j)}\right]\ge
\frac{(1-\e)C(B)}{R(B)^2}\Clust(A|B).
$$
Note that there is no difficulty to compute $\sigma$ efficiently
once the partition $\{A_1,\ldots,A_k\}$ is given, since these sets
are simplicial cones. The issue with efficiency here is how to
compute this partition in polynomial time. As we discussed in
Remark~\ref{rem:how to compute}, this can be done when $k$ is fixed
(or grows very slowly with $n$), but we do not know how to do this
when, say, $k=\sqrt{n}$.

\section{Matching Unique Games hardness}\label{sec:UGC}

In this section we show that for a fixed positive semi-definite
matrix $B$, approximating $\Clust(A|B)$ within a ratio strictly
smaller than $\frac{R(B)^2}{C(B)}$ is Unique Games hard. We will
study functions $f: \{1,\ldots,k\}^n \to \R$ and their Fourier
spectrum at the first level. A novel feature of our proof is that
our Fourier analysis will be carried out with respect to a
distribution on $\{1,\ldots,k\}$
  that is not necessarily uniform. In fact, the choice of the distribution itself
is dictated by the matrix $B$ as described in
Section~\ref{sec:choose-mu}.



\subsection{Choosing a special probability distribution on $\{1,\ldots,k\}$} \label{sec:choose-mu}

\begin{fact}\label{fact:separation}  Let $B=(b_{ij})$ be a $k \times k$ symmetric positive semi-definite matrix
and $b_{ij} = \langle v_i, v_j \rangle$ be its Gram representation,
where $v_1, \ldots, v_k$ are vectors (w.l.o.g.)  in $\R^k$.   Let
$R(B)$ be the minimum radius of a Euclidean ball containing all
these vectors, and $w(B)$ be the center of this ball. Then $w(B)$ is
a convex combination of the $v_i$'s that are on the boundary of the
ball. In other words, there exist non-negative coefficients $p(1),
\ldots, p(k)$ such that $\sum_{i=1}^k p(i) = 1$,
 $ \ w(B) = \sum_{i=1}^k p(i) v_i$ and  $p(i) \not= 0$
only if $\| v_i - w(B)\|_2 = R(B)$.
\end{fact}

Fact~\ref{fact:separation} is well known (see for example the proof
of Proposition 1.13 in~\cite{BenLin}). Its proof is a simple
separation argument. Indeed, define $J\coloneqq \{j\in
\{1,\ldots,k\}:\ \|v_j-w(B)\|_2=R(B)\}$ and let $K$ be the convex
hull of $\{v_j\}_{j\in J}$. Assume for the sake of contradiction
that $w(B)\notin K$. Then there would be a hyperplane $H$ separating
$w(B)$ from $K$. Moving $w(B)$ a little in the direction of $H$
would turn the equalities on $J$ to strict inequalities, while
preserving the strict inequalities off $J$. This contradicts the
minimality of $R(B)$.

We intend to use the probability distribution $(p(1), \ldots, p(k))$
from fact~\ref{fact:separation}.
However, for technical reasons,  we need the probability mass for
each atom to be non-zero, and therefore, we will use a very small
perturbation of this distribution. Towards this end we define
$\mu(i) = (1-\beta) p(i) + \frac{\beta}{k}$ for every $i\in
\{1,\ldots,k\}$. The value of $\beta
> 0$ is chosen to be sufficiently small as in the following lemma.

\begin{lemma}\label{lemma:rb-approx}  Fix any $\eps > 0$ and the matrix $B$. Then for
a sufficiently small $\beta = \beta(\eps, B) > 0$,
\begin{equation}
 \sum_{i=1}^k \mu(i) \left\| v_i - \sum_{j=1}^k \mu(j) v_j \right\|_2^2
\geq R(B)^2-\eps. \label{eqn:p-approx}
\end{equation}
\end{lemma}
\begin{proof} Note that if $\beta = 0$, then $\mu(i)=p(i)$ for all $i\in \{1,\ldots, k\}$,
 and
 $$\sum_{i=1}^k \mu(i) \left\| v_i - \sum_{j=1}^k \mu(j) v_j \right\|_2^2 =
 \sum_{i=1}^k p(i) \| v_i - w(B) \|_2^2 = R(B)^2,$$
 since  $p(i) \not= 0$
only if $\| v_i - w(B)\|_2 = R(B)$.  Thus by continuity for
sufficiently small $\beta$ the inequality (\ref{eqn:p-approx})
holds. For concreteness we also give a direct argument which gives a
reasonable bound on $\beta$. Assume that $\beta<\frac17$. Then,
using the fact that $\mu\ge (1-\beta)p$ (point-wise), we see that:
\begin{eqnarray*}
&&\!\!\!\!\!\!\!\!\!\!\!\!\!\!\!\left(\sum_{i=1}^k \mu(i) \left\|
v_i - \sum_{j=1}^k \mu(j) v_j \right\|_2^2\right)^{1/2}\ge
\sqrt{1-\beta}\left(\sum_{i=1}^k p(i) \left\|(1-\beta)\left( v_i -
\sum_{j=1}^k p(j) v_j\right)+\frac{\beta}{k}\sum_{j=1}^k(v_i-v_j)
\right\|_2^2\right)^{1/2}\\
&\ge& \sqrt{1-\beta}\left(\sum_{i=1}^k p(i) \left\|(1-\beta)
(v_i-w(B))\right\|_2^2\right)^{1/2}-\sqrt{1-\beta}\left(\sum_{i=1}^kp(i)
\left\|\frac{\beta}{k}\sum_{j=1}^k(v_i-v_j)\right\|_2^2
\right)^{1/2}\\&\ge&
(1-\beta)^{3/2}R(B)-\beta\sqrt{1-\beta}\left(\sum_{i=1}^kp(i)
\frac{1}{k}\sum_{j=1}^k\left\|v_i-v_j\right\|_2^2 \right)^{1/2}\\
&\ge& (1-\beta)^{3/2}R(B)-\beta\sqrt{1-\beta}\max_{i,j\in
\{1,\ldots,k\}}\|v_i-v_j\|_2\\
&\ge& \sqrt{1-\beta}\left(1-3\beta\right)R(B)\\&\ge&
\sqrt{1-7\beta}\cdot R(B),
\end{eqnarray*}
where in the penultimate inequality we used the trivial fact that
$\max_{i,j\in \{1,\ldots,k\}}\|v_i-v_j\|_2\le 2R(B)$. Thus we can
take $\beta=\frac{\e}{7R(B)^2}$ to ensure the validity
of~\eqref{eqn:p-approx}.
\end{proof}

Henceforth we fix the probability space $(\Omega = \{1,\ldots,k\},
\mu)$. 
Let $U=(u_{ij})$ be a $k\times k$ orthogonal matrix such that
$u_{1j}=\sqrt{\mu(j)}$ for all $j\in \{1,\ldots,k\}$ (such an
orthogonal matrix exists since this ensures that $\sum_{j=1}^k
u_{1j}^2=1$). Now define random variables
$X_1,\ldots,X_k:\{1,\ldots,k\}\to \R$ by
$X_i(j)=\frac{u_{ij}}{\sqrt{\mu(j)}}$ (here is one place where we
need the atoms of $\mu$ to have positive mass. We will also use this
fact to allow for the application of the result of~\cite{MOO} in the
proof of Theorem~\ref{thm:dict-soundness} below). Then by design
$X_1$ is the constant $1$ function, and for all $i,j\in
\{1,\ldots,k\}$ we have:
$$
\sum_{\ell=1}^k\mu(\ell)X_i(\ell)X_j(\ell)=\sum_{\ell=1}^k
u_{i\ell}u_{j\ell}=(UU^t)_{ij}=\delta_{ij},
$$
where $\delta_{ij}$ is the Kronecker delta. Similarly:
$$
\sum_{\ell=1}^kX_\ell(i)X_\ell(j)=\frac{1}{\sqrt{\mu(i)\mu(j)}}\sum_{\ell=1}^k
u_{\ell i}u_{\ell
j}=\frac{(U^tU)_{ij}}{\sqrt{\mu(i)\mu(j)}}=\frac{\delta_{ij}}{\mu(i)}.
$$
By relabeling these random variables (for the sake for simplicity of
later notation) we thus obtain the following lemma:
\begin{lemma} \label{lemma:ortho-basis}
There exist random variables $X_0, X_1,\ldots, X_{k-1}$ on $\Omega$
such that:
\begin{itemize}
\item $X_0 \equiv 1$.
\item For $i,j\in \{0,\ldots,k-1\}$ we have
$$\E_{\mu} [ X_i  X_j
 ]= \left\{\begin{array}{ll}0 & \mathrm{if}\ i\not=j,\\1 &\mathrm{if}\
 i=j.\end{array}\right.$$
\item For every $\omega, \omega' \in \Omega$ we have
$$\sum_{i=0}^{k-1} X_i(\omega) X_i(\omega')= \left\{\begin{array}{ll}0 & \mathrm{if}\ \omega\not=\omega',
\\\frac{1}{\mu(\omega)} &\mathrm{if}\
 \omega=\omega'.\end{array}\right.$$
\end{itemize}
\end{lemma}

\subsection{Dictatorships vs. functions with small influences}

In this section we
will associate to every function from $\{1,\ldots,k\}^n$ to
$$\Delta_k\coloneqq \left\{ x \in \R^k : \ x_i \geq 0 \
\forall \ i \in \{1,\ldots,k\}, \  \ \sum_{i=1}^k x_i = 1 \right\}$$
a numerical parameter, or ``objective value". We will show that the
value of this parameter for functions which depend only on a single
coordinate (i.e. dictatorships) differs markedly from its value on
functions which do not depend significantly on any particular
coordinate (i.e. functions with small influences). This step is an
analog of the ``dictatorship test" which is prevalent in PCP based
hardness proofs.

We begin with some notation and preliminaries on Fourier-type
expansions. For any function $f: \R^n \to \Delta_k$ we write
$f=(f_1, f_2, \ldots, f_k)$ where $f_i : \R^n \to [0,1]$ and
$\sum_{i=1}^k f_i = 1$. With this notation we have
 $$ C(B) =\sup_{f : \R^{k-1} \to \Delta_k}  \sum_{i=1}^k\sum_{j=1}^k b_{ij}
   \left\langle  \int_{\R^{k-1}} x
 f_i(x) d\gamma_{k-1} (x), \ \int_{\R^{k-1}} x
 f_j(x) d\gamma_{k-1} (x)   \right\rangle $$
where $C(B)$ is as in Section~\ref{sec:simplices}.  We have already
seen that the supremum above is actually attained. Also $C(B)$
remains the same if the supremum is taken over functions over $\R^n$
with $n \geq k-1$, i.e. for every $n \geq k-1$,
$$ C(B) =\sup_{f : \R^{n} \to \Delta_k}  \sum_{i=1}^k\sum_{j=1}^k b_{ij}
   \left\langle  \int_{\R^{n}} x
 f_i(x) d\gamma_{n} (x), \ \int_{\R^{n}} x
 f_j(x) d\gamma_{n} (x)   \right\rangle.  $$

Let $(\Omega = \{1,\ldots,k\}, \mu)$ be the probability space as
chosen in Section \ref{sec:choose-mu}.
  Let $(\Omega^n, \mu^n)$ be the associated product space. We
will be analyzing functions $f: \Omega^n \to \Delta_k$ (and more
generally into $\R^k$). As in Lemma \ref{lemma:ortho-basis}, fix a
basis of orthonormal random variables on $\Omega$ where one of them
is the constant $1$ function, that is  $\{X_0 \equiv 1, X_1, \ldots,
X_{k-1}\}$.
 Then any
function $f: \Omega \to \R$ can be written as a linear combination
of the $X_i$'s.

In order to analyze functions $f: \Omega^n \to \R$, we let $\calX =
 (\calX_1, \calX_2, \ldots, \calX_n)$ be an ``ensemble" of random variables
 where for $i\in \{1,\ldots,n\}$ we write $\calX_i = \{ X_{i,0}, X_{i,1}, \ldots, X_{i,k-1}\}$, and
 for every $i$,  $\{X_{i,j}\}_{j=0}^{k-1}$ are independent copies of  the $\{X_j\}_{j=0}^{k-1}$. Any
 $\sigma = (\sigma_1, \sigma_2, \ldots, \sigma_n) \in \{0,1,2,\ldots, k-1\}^n$ will be called a
 multi-index. We shall denote by $|\sigma|$ the number on non-zero
 entries in $\sigma$.
 Each multi-index defines a monomial
 $$
 x_{\sigma}
  := \prod_{\substack{i \in \{1,\ldots,n\}\\ \sigma_i\not=0}} x_{i, \sigma_i}
  $$ on a set of $n(k-1)$ indeterminates $\{x_{ij} \ | \
   i \in \{1,\ldots,n\}, j \in \{1,2,\ldots, k-1\} \}$, and also a random variable
   $X_\sigma : \Omega^n \to \R$ as
    $$ X_\sigma (\omega) :=  \prod_{i=1}^n X_{i,\sigma_i} (\omega_i). $$
The random variables $\{X_\sigma\}_\sigma$ form an orthonormal basis
for the space of functions $f: \Omega^n \to \R$. Thus, every such
$f$ can be written uniquely as (the ``Fourier expansion")
 $$ f = \sum_{\sigma} \widehat{f}(\sigma) X_\sigma, \quad \widehat{f}(\sigma)
 \in \R. $$
We denote the corresponding multi-linear polynomial as $Q_f =
\sum_\sigma \widehat{f}(\sigma) x_\sigma$. One can think of $f$ as
the polynomial $Q_f$ applied to the ensemble $\calX$, i.e. $f =
Q_f(\calX)$. Of course, one can also apply $Q_f$ to any other
ensemble, and specifically to the Gaussian ensemble $\calG =
(\calG_1, \calG_2, \ldots, \calG_n)$ where $\calG_i = \{
G_{i,0}\equiv 1, G_{i,1},\ldots, G_{i, k-1}\}$ and $G_{i,j}, i \in
\{1,\ldots,n\}, j\in \{1,\ldots,k-1\}$ are i.i.d. standard
Gaussians. Define the influence of the $i$'th variable on $f$ as
$$
\infl_i(f)\coloneqq \sum_{\sigma_i\neq 0} \widehat f(\sigma)^2.
$$
Roughly speaking, the results of~\cite{Rotar79,MOO}
 say that if $f: \Omega^n \to [0,1]$ is a
function  all of whose influences are small, then $f= Q_f(\calX)$
and $Q_f(\calG)$ are almost identically distributed, and in
particular, the values of $Q_f(\calG)$ are essentially contained in
 $[0,1]$. Note that $Q_f(\calG)$ is a random variable on
the probability space $(\R^{n (k-1)}, \gamma_{n(k-1)})$.



Consider functions $f: \Omega^n \to \Delta_k$. We write $f= (f_1,
f_2, \ldots, f_k)$ where $f_i : \Omega^n \to [0,1]$ with
$\sum_{i=1}^k f_i = 1$. Each $f_i$ has a unique representation
(along with the corresponding multi-linear polynomial)
 $$f_i = \sum_\sigma \widehat{f_i}(\sigma) X_\sigma,  \quad \quad Q_i := Q_{f_i} =
  \sum_{\sigma} \widehat{f_i}(\sigma) x_\sigma.  $$

We shall define an objective function ${\rm OBJ}(f)$ that is a
positive semidefinite quadratic form on the table of values of $f$
which corresponds to a centered symmetric positive semidefinite
bilinear form. Then we analyze the value of this objective function
when $f$ is a ``dictatorship" versus when $f$ has all low
influences.

\subsubsection*{The objective value} For a function $f: \Omega^n \to
\Delta_k$ (or more generally, $f: \Omega^n \to \R^k$) define
\begin{equation}\label{eq:OBJ}
 \mbox{OBJ} (f)   :=  \sum_{i=1}^k\sum_{j=1}^k b_{ij}
  \left( \sum_{\sigma:\  |\sigma|=1} \widehat{f_i}( \sigma) \widehat{f_j}( \sigma) \right).
\end{equation}
 Note that there are
$n(k-1)$ multi-indices $\sigma$ such that $|\sigma|=1$.

\subsubsection*{The objective value for dictatorships} For $\ell \in
\{1,\ldots,n\}$ we define a dictatorship function $f^{dict, \ell}:
\Omega^n \to \Delta_k$ as follows. The range of the function is
limited to only $k$ points in $\Delta_k$, namely the points $\{e_1,
e_2, \ldots, e_k\}$ where $e_i$ is a vector with $i^{th}$ coordinate
$1$ and all other coordinates zero.

\begin{equation}
  f^{dict, \ell} (\omega)  :=  e_i \ \quad  \mbox{if}  \ \omega_\ell = i.
\end{equation}
In other words, when one writes
 $ f^{dict, \ell}  = (f_1, f_2, \ldots, f_k)$, for $i\in \{1,\ldots,k\}$,
 $f_i$ is  $\{0,1\}$-valued and
 $f_i(\omega) =1$ if and only if $\omega_\ell = i$.
 The Fourier expansion of $f_i$
  is
 \begin{eqnarray}\label{eq:expand dict}
 f_i (\omega) =  \mu(i) \sum_{\sigma:\  \sigma_j=0 \ \forall j \not= \ell}
  X_{\sigma_\ell}(i) \    X_\sigma (\omega).
   \end{eqnarray}
Indeed, the right hand side of~\eqref{eq:expand dict} equals
\begin{eqnarray*}
 \mu(i) \sum_{0 \leq \sigma_\ell \leq k-1}
  X_{\sigma_\ell}(i) \    X_{\sigma_\ell} (\omega_\ell) =
  \left\{ \begin{array}{l}
           1 \ \  \mbox{if} \ \omega_\ell = i, \\
           0 \ \  \mbox{otherwise.} \quad\quad\quad \mbox{(see Lemma \ref{lemma:ortho-basis})}
           \end{array} \right.
\end{eqnarray*}
Thus,
\begin{eqnarray}
\mbox{OBJ} \left(f^{dict,\ell}\right)   & = &
\sum_{i=1}^k\sum_{j=1}^k b_{ij}
  \left( \sum_{\sigma:\  |\sigma|=1} \widehat{f_i}( \sigma) \widehat{f_j}( \sigma) \right) \nonumber \\
  & = & \sum_{i=1}^k\sum_{j=1}^k b_{ij} \left( \sum_{r=1}^{k-1} \mu(i)X_r(i) \mu(j) X_r(j) \right) \nonumber \\
  & = & \sum_{i=1}^k\sum_{j=1}^k b_{ij} \cdot \mu(i)\mu(j) \left(\sum_{r=0}^{k-1} X_r(i) X_r(j) - 1 \right) \nonumber \\
  & = & \sum_{\substack{i,j\in \{1,\ldots,k\}\\i\neq j}} \langle v_i, v_j \rangle \cdot \mu(i)\mu(j) (-1)
   + \sum_{i=1}^k \langle v_i, v_i \rangle \cdot \mu(i)^2 \left( \frac{1}{\mu(i)} -1  \right) \nonumber \\
  & = & \sum_{i=1}^k \mu(i) \left\| v_i - \sum_{j=1}^k \mu(j)v_j \right\|_2^2 \nonumber \\
  & \geq & R(B)^2- \eps, \label{eqn:rbs-dict}
\end{eqnarray}
using Lemma \ref{lemma:rb-approx}.

\subsubsection*{The objective value for functions with low
influences} For $f:\Omega^n\to \R$, $j\in \{1,\ldots,n\}$ and $m\in
\mathbb N$ denote (the ``degree $m$-influence" of $f$):
$$
\infl_j^{\le m}(f)\coloneqq \sum_{\substack{|\sigma|\le
m\\\sigma_j\neq 0}}\widehat f(\sigma)^2.
$$
For every $0 \leq \rho \leq 1$ we will use the smoothing operator:
$$
T_\rho f = \sum_{\sigma} \rho^{|\sigma|}\widehat{f}(\sigma)
X_\sigma.
$$

Equivalently,
$$ T_\rho f (\omega_1,\ldots,\omega_n) = \E [f(\omega'_1,\ldots,\omega'_n)], $$
where independently for each $i$, $\omega'_i$ is chosen to be
$\omega_i$ with probability $\rho$ and a random (with respect to the
underlying distribution $\mu$) element in $\Omega$ with probability
$1-\rho$.

\medskip
The following theorem is the key analytic fact used in our UGC
hardness result:

\begin{theorem}\label{thm:dict-soundness}
 For every $\eps > 0$, there exists $\tau > 0$  so that the following holds: for any function
 $f : \Omega^n\to \Delta_k$ which satisfies
 $$\forall\  i \in \{1,\ldots,k\},\ \forall\  j \in \{1,\ldots,n\},  \quad \infl_j^{\leq \log (1/\tau)}(f_i) \leq \tau
 $$
 we have,
 $$  {\rm OBJ}(f) \leq C(B)+\eps. $$
\end{theorem}

\begin{proof} Let $\delta, \eta > 0$ be sufficiently small constants to be chosen later.
Let   $Q_i= Q_{f_i}$ be the multi-linear polynomial associated with
$f_i$. Recall that $Q_i$ is a multi-linear polynomial in the
$n(k-1)$ indeterminates
 $\left\{ x_{jp} \ | \ j \in \{1,\ldots,n\}, p \in \{1,\ldots, k-1\}\right\}$.  Moreover $f_i = Q_i(\calX)$
has range $[0,1]$ and $\sum_{i=1}^k f_i = 1$.


Let $R_i = (T_{1-\delta} Q_i) (\calX) $ and $S_i = (T_{1-\delta}
Q_i) (\calG)$ (the smoothening operator $T_{1-\delta}$ helps  us
meet some technical pre-conditions before applying the invariance
principle of~\cite{MOO}).
 Note that $R_i$ has range $[0,1]$ and $S_i$ has range $\R$. It will follow however from~\cite{MOO}
 that $S_i$ is {\it essentially}  in $[0,1]$. First we relate
${\rm OBJ}(f)$ to the functions $S_i$ which will, up to truncation,
induce a partition of $\R^{n(k-1)}$, which in turn will give the
bound in terms of $C(B)$.



\begin{eqnarray}
&&\!\!\!\!\!\!\!\!\!\!\!\!\!\!\!\!(1-\delta)^2 \cdot \mbox{OBJ}(f)
    =  (1-\delta)^2 \sum_{i=1}^k\sum_{\ell=1}^k b_{i\ell} \sum_{\sigma: |\sigma|=1} \widehat{f_i}(\sigma)\widehat
    {f_\ell}(\sigma)  \nonumber \\
   & = & (1-\delta)^2 \sum_{i=1}^k\sum_{\ell=1}^k b_{i\ell}  \sum_{j=1}^n \sum_{p=1}^{k-1}
    \left(  \int_{\R^{n(k-1)}}
  x_{jp}  \ Q_i (x) d\gamma_{n(k-1)} (x)\right) \cdot \left(   \int_{\R^{n(k-1)}}
  x_{jp}  \ Q_\ell (x) d\gamma_{n(k-1)} (x)  \right)  \nonumber \\
  & = & (1-\delta)^2 \sum_{i=1}^k\sum_{\ell=1}^k  b_{i\ell}
    \left\langle  \int_{\R^{n(k-1)}}
  x   \ Q_i (x) d\gamma_{n(k-1)} (x),   \int_{\R^{n(k-1)}}
  x  \ Q_\ell (x) d\gamma_{n(k-1)} (x)  \right\rangle  \nonumber \\
 & = & \sum_{i=1}^k\sum_{\ell=1}^k b_{i\ell}
     \left\langle  \int_{\R^{n(k-1)}}
  x   \ (T_{1-\delta}Q_i) (x) d\gamma_{n(k-1)} (x),   \int_{\R^{n(k-1)}}
  x  \ (T_{1-\delta}Q_\ell)  (x) d\gamma_{n(k-1)} (x)  \right\rangle  \nonumber \\
   & = & \sum_{i=1}^k\sum_{\ell=1}^k b_{i\ell}
     \left\langle  \int_{\R^{n(k-1)}}
  x   \ S_i (x) d\gamma_{n(k-1)} (x),   \int_{\R^{n(k-1)}}
  x  \ S_\ell  (x) d\gamma_{n(k-1)} (x)  \right\rangle.
    \label{eqn:estimate-Rn}
\end{eqnarray}
We shall now bound the last term above by $C(B)+ o(1)$. For any
real-valued function $h$ on $\R^{n (k-1)}$, let
   $$ \chop(h) (x) :=   \left\{ \begin{array}{ll}
                         0 & \mbox{if} \ h(x) < 0, \\
                         h(x) & \mbox{if} \ h(x) \in [0,1], \\
                         1 & \mbox{if} \ h(x) > 1.
                           \end{array}  \right.  $$
  Applying Theorem 3.20 in~\cite{MOO} to the polynomial $Q_i$, it
follows that (provided $\tau$ is sufficiently small compared to
$\delta$ and $\eta$),
 \begin{eqnarray} \label{error-term-1}
 \left\| S_i  - \chop(S_i) \right\|_{L_2(\gamma_{n(k-1)})}^2 =  \int_{\R^{n(k-1)}}
 \left|S_i(x) - \chop(S_i)(x) \right|^2 d\gamma_{n(k-1)} (x) \leq \eta.
 \end{eqnarray}

The functions $\chop(S_i)$ are almost what we want except that they
might not sum up to  $1$. So further define
\begin{equation}
 S^*_i(x) := \frac{\chop(S_i)(x)}{\sum_{i=1}^k  \chop(S_i) (x)}.
                      \nonumber
\end{equation}
Clearly, $\left\{S^*_i\right\}_{i=1}^k$ have range $[0,1]$ and
$\sum_{i=1}^k S^*_i \equiv 1$. Observe that the following holds
point-wise:

$$  \sum_{j=1}^k  \left|\chop(S_j) - S_j^* \right|
 = \left| \sum_{j=1}^k \chop(S_j)-1 \right|  =  \left| \sum_{j=1}^k \chop(S_j)-
  \sum_{j=1}^k S_j  \right| \leq \sum_{j=1}^k \left| S_j -
  \chop(S_j)\right|,
   $$
where we used that $\sum_{j=1}^k S_j = T_{1-\delta} \sum_{j=1}^k Q_j
= T_{1-\delta} 1 = 1$.  It follows that for all $i\in
\{1,\ldots,k\}$ we have:
$$\left\|\chop(S_i) - S_i^*\right\|_{L_2(\gamma_{n(k-1)})}\le\sum_{j=1}^k\left\|\chop(S_j) - S_j^*\right\|_{L_2(\gamma_{n(k-1)})} \leq \sum_{j=1}^k \left\|S_j -
\chop(S_j) \right\|_{L_2(\gamma_{n(k-1)})}
 \leq k \sqrt{\eta},
 $$
 where we used (\ref{error-term-1}).  Finally,

\begin{equation}
\left\| S_i - S^*_i \right\|_{L_2(\gamma_{n(k-1)})} \leq \left\|S_i
- \chop(S_i) \right\|_{L_2(\gamma_{n(k-1)})} + \left\|\chop(S_i) -
S^*_i\right\|_{L_2(\gamma_{n(k-1)})} \leq
  (k+1)\sqrt{\eta}.   \label{eqn:error-term-2}
\end{equation}   Now write
\begin{eqnarray}\label{eq:two integrals} u_i = \int_{\R^{n(k-1)}}
  x \ S_i (x) d\gamma_{n(k-1)} (x), \quad \quad  w_i  =  \int_{\R^{n(k-1)}}
  x \ S^*_i (x) d\gamma_{n(k-1)} (x). \end{eqnarray}
The norm of $u_i - w_i$ is bounded by $(k+1)\sqrt{\eta}$ using
 (\ref{eqn:error-term-2}) and Lemma
\ref{lemma:trivial-bound} below.  Since $| S^*_i | \leq 1$, the norm
of $w_i$ is bounded by $1$.  Returning to the estimation in Equation
(\ref{eqn:estimate-Rn}) and applying Lemma \ref{lemma:uw-approx}
below, we see that:

\begin{equation*}
(1-\delta)^2  \cdot \mbox{OBJ}(f)  =  \sum_{i=1}^k\sum_{\ell=1}^k
b_{i\ell} \langle u_i, u_\ell\rangle
  \leq
 \sum_{i=1}^k\sum_{\ell=1}^k b_{i\ell} \langle w_i, w_\ell\rangle +
   O\left(k \sqrt{\eta}\right)\left( \sum_{i=1}^k\sum_{\ell=1}^k |b_{i\ell}|\right).
\end{equation*}
Since $\sum_{i=1}^k S_i^* \equiv 1$ we have
\begin{multline*}
 \sum_{i=1}^k\sum_{\ell=1}^k b_{i\ell} \langle w_i, w_\ell\rangle
  =  \sum_{i=1}^k\sum_{\ell=1}^k b_{i\ell}\left\langle  \int_{\R^{n(k-1)}}
  x \ S_i^* (x) d\gamma_{n(k-1)} (x), \ \int_{\R^{n(k-1)}}
  x \ S_\ell^* (x) d\gamma_{n(k-1)} (x)   \right\rangle \\
     \leq  \sup_{f: \R^{n(k-1)} \to \Delta_k}
  \left( \sum_{i=1}^k\sum_{\ell=1}^k b_{i\ell}
   \left\langle  \int_{\R^{n(k-1)}}
  x \ f_i (x) d\gamma_{n(k-1)} (x), \ \int_{\R^{n(k-1)}}
  x \ f_\ell (x) d\gamma_{n(k-1)} (x)   \right\rangle   \right)    = C(B).
\end{multline*}
It follows that $\mbox{OBJ}(f)
 \leq C(B) + \eps$, provided that $\eta$ and $\delta$ are small
 enough.
\end{proof}

\begin{lemma}\label{lemma:trivial-bound}
 Let $g \in L_2(\R^n, \gamma_n)$. Then
 $$ \left\| \int_{\R^n} x \ g(x) d\gamma_n (x) \right\|_2  \leq  \|g \|_{L_2(\R^n, \gamma_n)}. $$
\end{lemma}
\begin{proof} Note that the square of the left hand side  equals
 $$ \sum_{i=1}^n \left| \int_{\R^n} x_i \ g(x)  d\gamma_n (x) \right|^2
  =  \sum_{i=1}^n  \langle x_i, g \rangle^2.  $$
Since $x_i \in  L_2(\R^n, \gamma_n)$ are an orthonormal set of
functions, the sum of squares of projections of $g$ onto them is at
most the squared norm of $g$.
\end{proof}

\begin{lemma}\label{lemma:uw-approx}
 Suppose $\{u_i\}_{i=1}^k$ and $\{w_i\}_{i=1}^k$ are vectors in $\R^n$ such
 that $\|u_i - w_i \|_2 \leq d$
for every $i\in \{1,\ldots,k\}$ and $\|w_i \|_2 \leq 1$. Let
$B=(b_{ij})$ be a $k\times k$  matrix. Then
 $$ \left|
  \sum_{i=1}^k\sum_{\ell=1}^k b_{i\ell} \langle u_i, u_\ell \rangle  -
    \sum_{i=1}^k\sum_{\ell=1}^k b_{i\ell} \langle w_i, w_\ell \rangle \right| \leq
    \left(2d+d^2\right)\sum_{i=1}^k\sum_{\ell=1}^k |b_{i\ell}|.$$
\end{lemma}
\begin{proof} From the given conditions on the norms of $a_i = u_i-w_i$ and $w_i$,
it follows that for any $ i,\ell\in\{1,\ldots,k\}$,
 $$ \left| \langle u_i, u_\ell \rangle - \langle w_i , w_\ell \rangle \right|
  \leq |\langle a_i, w_\ell \rangle| +
  |\langle a_\ell, w_i \rangle| +
  |\langle a_i, a_\ell \rangle| \leq    2d+d^2. $$
Hence,
$$ \left|
  \sum_{i=1}^k\sum_{\ell=1}^k b_{i\ell} \langle u_i, u_\ell \rangle  -
    \sum_{i=1}^k\sum_{\ell=1}^k b_{i\ell} \langle w_i, w_\ell \rangle \right| \leq
 \sum_{i=1}^k\sum_{\ell=1}^k |b_{i\ell}| \left| \langle u_i, u_\ell \rangle  -
    \langle w_i, w_\ell \rangle \right| \leq
    \left(2d+d^2\right)  \sum_{i=1}^k\sum_{\ell=1}^k |b_{i\ell}|,$$
    as required.
\end{proof}

\subsubsection*{The intended hardness factor}

As we show next, the dictatorship test can be translated (in a more
or less standard way by now) into a Unique Games hardness result.
The hardness factor (as usual) turns out to be
 the ratio of the objective value when the function is a
dictatorship versus when the function has all low influences, i.e.
$$\frac{R(B)^2-\eps}{C(B)+\eps}
 = \frac{R(B)^2}{C(B)} - o(1).
 $$

\subsection{The reduction from unique games to kernel clustering}

Given a Unique Games Instance $\calL(G(V,W,E), n,
\{\pi_{vw}\}_{(v,w)\in E})$, we construct an instance  of the
clustering problem.

\subsubsection*{Reformulation of the clustering problem}

As in our earlier paper \cite{KN08},  we first reformulate
 the kernel clustering problem for the ease of presentation.  As
 observed there, we can
reformulate it as (the matrix $A$ in the problem $\Clust(A|B)$ is
captured by the quadratic form $Q$ below):

\medskip
\noindent {\bf Kernel Clustering Problem:} Given a $k\times k$
symmetric positive semidefinite matrix B, and a symmetric positive
semidefinite quadratic form $Q(\cdot ,\cdot )$ on $\R^N \times
\R^N$, find
 $F: \{1,\ldots,N\} \to \Delta_k$,  $F= (F_1, F_2,\ldots, F_k)$, so as to maximize
 $\sum_{i=1}^k\sum_{j=1}^k b_{ij} Q(F_i, F_j)$.

\subsubsection*{The clustering problem instance} Given a Unique
Games instance $\L\left(G(V,W,E), n, \{\pi_{vw}\}_{(v,w)\in
E}\right)$, the clustering  problem is to find a function $F: W
\times \Omega^n \to \Delta_k$ so as to maximize
$\sum_{i=1}^k\sum_{j=1}^k b_{ij}Q(F_i, F_i)$ where $Q$ is a suitably
defined symmetric positive semidefinite quadratic form.  For
notational convenience, we write:
   $$ F_w :=  F(w, \cdot),  \quad\quad F_w : \Omega^n \to \Delta_k.$$
Also, for every $v \in V$, we write:
$$  F_v \coloneqq  \E_{(v,w) \in E} \left[ F_w  \circ \pi_{vw} \right], \quad\quad\quad
  F_v : \Omega^n \to \Delta_k.  $$
We used the following notation: for any function $g : \Omega^n \to
\Delta_k$ and $\pi : \{1,\ldots,n\} \to \{1,\ldots,n\}$ we write  $g
\circ \pi: \Omega^n \to \Delta_k$ for the  function $ (g \circ \pi)
(\omega) :=  g ( \omega_{\pi(1)}, \omega_{\pi(2)}, \ldots,
\omega_{\pi(n)} )$. As usual, we denote $F_w = (F_{w,1}, F_{w,2},
\ldots, F_{w,k})$ where each $F_{w,i}$ has range $[0,1]$ and
$\sum_{i=1}^k F_{w,i} = 1$. Similarly, $F_v = (F_{v,1}, F_{v,2},
\ldots, F_{v,k})$ and $\sum_{i=1}^k F_{v,i} = 1$. Now we are ready
to define the clustering problem instance.

\medskip\noindent {\bf Clustering instance:} The goal is to find
$F:  W \times \Omega^n \to \Delta_k$ so as to maximize:
\begin{eqnarray}\label{eqn:clustering-obj}
  \max_{F: W \times \Omega^n \to \Delta_k} \ \
   {\E}_{v \in V} \left[ \mbox{OBJ} (F_v) \right] =
    \max_{F: W \times \Omega^n \to \Delta_k} \ \
    \ \E_{v \in V}
   \left[ \sum_{i=1}^k\sum_{j=1}^k b_{ij} \sum_{\sigma: |\sigma|=1} \widehat{F}_{v,i}(\sigma)
     \cdot \widehat{F}_{v,j}(\sigma)
   \right].
\end{eqnarray}

\subsubsection*{Completeness} We will show that if the Unique
Games instance has an almost satisfying labeling, then the objective
value of the clustering problem is at least $R(B)^2-o(1)$. So, let
$\rho: V \cup W \to \{1,\ldots,n\}$ be the labeling, such that for
at least $1-\eps$ fraction of the vertices $v \in V$ (call such $v$
good) we have
 $$  \pi_{vw} ( \rho(w) ) = \rho(v) \ \ \forall \ (v,w) \in E. $$
Define $F: W \times \Omega^n \to \Delta_k$ as follows: for every $w
\in W$,
 $F_w: \Omega^n \to \Delta_k$
  equals the dictatorship corresponding to $\rho(w) \in \{1,\ldots,n\}$, i.e.,
   $$ F_w :=   f^{dict, \rho(w)}.$$

\begin{lemma}[\cite{KN08}]
For a good $v\in V$ we have   $F_v = f^{dict, \rho(v)}$.
\end{lemma}

Thus the contribution of $v$ in  (\ref{eqn:clustering-obj}) is
$\mbox{OBJ}(f^{dict, \rho(v)}) \geq R(B)^2 - \eps$ as observed in
Equation (\ref{eqn:rbs-dict}). Since $1-\eps$ fraction of $v \in V$
are good, (\ref{eqn:clustering-obj}) is at least $(1-\eps)\cdot
(R(B)^2-\eps) = R(B)^2 - o(1)$.

\subsubsection*{Soundness} Suppose for the sake of contradiction that the value
of (\ref{eqn:clustering-obj}) is at least $C(B) + 2\eps$. As in
\cite{KN08}, it can be proved that  the Unique Games instance must
have a labeling that satisfies at least a constant fraction of its
edges, the constant depending on the parameter $\tau$ used in
Theorem \ref{thm:dict-soundness}. This is a contradiction, provided
the soundness of the Unique Games instance  is chosen to be even
lower to begin with. The proof is the same as in \cite{KN08}, by
replacing the $C(k)$ therein by $C(B)$ (\cite{KN08} focused on the
case when $B$ is the $k \times k$ identity matrix. The constant
$C(k)$ therein is same as our constant $C(B)$ when $B$ is the $k
\times k$ identity matrix).

\section{A concrete example}\label{sec:example}

In this section we will use our results to evaluate the UGC hardness
threshold of the problem of computing
\begin{equation}\label{eq:c}
\Clust\left(A\left|\  \begin{pmatrix}
  1 & 0 & 0 \\
   0 & 1 & 0\\
   0 & 0 & c
   \end{pmatrix}\right.\right),
\end{equation}
where $A\in M_n(\R)$ is centered, symmetric and positive
semidefinite and $c\in (0,\infty)$ is a parameter. The case $c=1$,
corresponding to $B=I_3$ (the $3\times 3$ identity matrix) was
evaluated in~\cite{KN08}, where it was shown that the UGC hardness
threshold in this case equals $\frac{16\pi}{27}$.

For general $c>0$ the optimization problem in~\eqref{eq:c}
corresponds to the following question: given $n$ random variables
$X_1,\ldots,X_n$ the goal is to partition them into three sets
$S_1,S_2,S_3\subseteq \{1,\ldots,n\}$ such that
\begin{equation}\label{eq:weighted c}
\sum_{i,j\in S_1} \E\left[X_iX_j\right]+\sum_{i,j\in S_2}
\E\left[X_iX_j\right]+c\sum_{i,j\in S_3} \E\left[X_iX_j\right]
\end{equation}
is maximized. Thus we wish to cluster the variables into three
clusters so as to maximize the intra-cluster correlations, while the
parameter $c$ allows us to tune the relative importance of one of
the clusters. We stress that we do not claim that this optimization
problem is of particular intrinsic importance. We chose it as a way
to concretely demonstrate our results for the simplest possible
perturbation of the case of $B=I_3$.
We remark that it is also possible to explicitly solve the case of
general $3\times 3$ diagonal matrices $B$, i.e., the case of a
general weighting of the clusters in~\eqref{eq:weighted c}. The
formula for the UGC hardness threshold for general $3\times 3$
diagonal matrices turns out to be quite complicated, so we chose to
deal only with~\eqref{eq:c} as a simple example for the sake of
illustration. Note that for $3\times 3$ matrices the
characterization of $C(B)$ in terms of planar conical partitions is
particularly simple, and allows for explicit computations of the UGC
hardness threshold in additional cases.

Denote $ B\coloneqq \begin{pmatrix}
  1 & 0 & 0 \\
   0 & 1 & 0\\
   0 & 0 & c
   \end{pmatrix}=(\langle v_i,v_j\rangle)_{i,j=1}^3$, where
$v_1=(1,0,0),v_2=(0,1,0),v_3=(0,0,\sqrt{c})\in \R^3$. The side
lengths of the triangle whose vertices are $v_1,v_2,v_3$ are
$\left\{\ell_1=\sqrt{1+c},\ell_2=\sqrt{1+c},\ell_3=\sqrt{2}\right\}$.
Note that this is an acute triangle, so its smallest bounding circle
coincides with its circumcircle, and therefore its radius is given
by~\cite{Johnson60}:
\begin{equation}\label{eq:radius}
R(B)^2=\frac{\ell_1^2\ell_2^2\ell_3^2}{(\ell_1+\ell_2+\ell_3)
(-\ell_1+\ell_2+\ell_3)(\ell_1-\ell_2+\ell_3)(\ell_1+\ell_2-\ell_3)}=\frac{(1+c)^2}{2+4c}.
\end{equation}

We shall now compute $C(B)$. By Lemma~\ref{lem:simplicial} the
partition $\{A_1,A_2,A_3\}$ of $\R^2$ at which $C(B)$ is attained
consists of disjoint cones of angles $\alpha_1,\alpha_2,\alpha_3\in
[0,2\pi]$ where $\alpha_1+\alpha_2+\alpha_3=2\pi$. A direct
computation shows that for $j\in \{1,2,3\}$ we have:
$$
\left\|\int_{A_j}xd\gamma_2(x)\right\|_2^2=\frac{1}{2\pi}\sin^2\left(\frac{\alpha_j}{2}\right).
$$
Hence
\begin{equation}\label{eq:maximum}
C(B)=\frac{1}{2\pi}\max_{\substack{\alpha_1,\alpha_2,\alpha_3\in
[0,2\pi]\\\alpha_1+\alpha_2+\alpha_3=2\pi}}\left(\sin^2\left(\frac{\alpha_1}{2}\right)
+\sin^2\left(\frac{\alpha_2}{2}\right)+c\sin^2\left(\frac{\alpha_3}{2}\right)
\right).
\end{equation}
Assume for the moment that the maximum in~\eqref{eq:maximum} is
attained when $\alpha_1,\alpha_2,\alpha_3\in (0,2\pi)$. Then using
Lagrange multipliers we see that
$\sin\alpha_1=\sin\alpha_2=c\sin\alpha_3$. This implies in
particular that either $\alpha_1=\alpha_2$ or (since
$\alpha_1,\alpha_2,\alpha_3\in (0,2\pi)$ and
$\alpha_1+\alpha_2+\alpha_3=2\pi$) $\alpha_1+\alpha_2=\pi$. In the
latter case $\alpha_3=\pi$, and it follows from the Lagrange
multiplier equations that $\sin\alpha_1=\sin\alpha_2=0$, which
forces one of $\{\alpha_1,\alpha_2\}$ to vanish, contrary to our
assumption. Hence we know that $\alpha_1=\alpha_2\coloneqq\alpha$.
Then $\alpha_3=2\pi-2\alpha$, and since $\alpha_3\in (0,2\pi)$ we
also know that $\alpha\in (0,\pi)$. The Lagrange multiplier
equations imply that $\sin
\alpha=c\sin(2\pi-2\alpha)=-2c\sin\alpha\cos\alpha$. Thus
$\cos\alpha= -\frac{1}{2c}$, and in particular we see that
necessarily $c\ge \frac12$. It follows that
$$
\sin^2\left(\frac{\alpha}{2}\right)=\frac{1-\cos\alpha}{2}=\frac{2c+1}{4c},
$$
and
$$
\sin^2\left(\frac{\alpha_3}{2}\right)=\sin^2\left(\pi-\alpha\right)=1-\cos^2\alpha=1-\frac{1}{4c^2}.
$$
Hence in this case:
\begin{equation}\label{eq:candidate}
\sin^2\left(\frac{\alpha_1}{2}\right)
+\sin^2\left(\frac{\alpha_2}{2}\right)+c\sin^2\left(\frac{\alpha_3}{2}\right)=2\frac{2c+1}{4c}+c\frac{4c^2-1}{4c^2}
=\frac{(2c+1)^2}{4c}.
\end{equation}

It remains to deal with the boundary case
$\{\alpha_1,\alpha_2,\alpha_3\}\cap\{0,2\pi\}\neq \emptyset$, which
as we have seen above is where the maximum in~\eqref{eq:maximum} is
necessarily attained if $c<\frac12$. If one of
$\{\alpha_1,\alpha_2,\alpha_3\}$ equals $2\pi$ then the others must
vanish, in which case $\sin^2\left(\frac{\alpha_1}{2}\right)
+\sin^2\left(\frac{\alpha_2}{2}\right)+c\sin^2\left(\frac{\alpha_3}{2}\right)=0$.
If one of $\{\alpha_1,\alpha_2,\alpha_3\}$ vanishes then in order to
maximize $\sin^2\left(\frac{\alpha_1}{2}\right)
+\sin^2\left(\frac{\alpha_2}{2}\right)+c\sin^2\left(\frac{\alpha_3}{2}\right)$
the other two must equal $\pi$, in which case the maximum value of
this quantity is $\max\{2,1+c\}$. Since $\max\{2,1+c\}$ never
exceeds the quantity $\frac{(2c+1)^2}{4c}$ from~\eqref{eq:candidate}
it follows that the maximum of
$\sin^2\left(\frac{\alpha_1}{2}\right)
+\sin^2\left(\frac{\alpha_2}{2}\right)+c\sin^2\left(\frac{\alpha_3}{2}\right)$
over $\{\alpha_1+\alpha_2+\alpha_3=2\pi\ \wedge\
\alpha_1,\alpha_2,\alpha_3\in [0,2\pi]\}$ equals
$\frac{(2c+1)^2}{4c}$ when $c\ge \frac12$ and equals $2$ when $c\le
\frac12$. We therefore proved that
\begin{equation}\label{eq:Cc}
C(B)=\left\{\begin{array}{ll}\frac{(2c+1)^2}{8\pi c}&\mathrm{if}\ c\ge\frac12,\\
\frac{1}{\pi}&\mathrm{if}\ c\le\frac12.\end{array}\right.
\end{equation}
By combining~\eqref{eq:radius} with~\eqref{eq:Cc} we conclude that
the UGC hardness threshold for computing~\eqref{eq:c} is:
\begin{eqnarray}\label{eq:phase}
\frac{R(B)^2}{C(B)}= \left\{\begin{array}{ll}\frac{4\pi
c(1+c)^2}{(1+2c)^3}&\mathrm{if}\ c\ge\frac12,\\
\frac{\pi (1+c)^2}{2+4c}&\mathrm{if}\ c\le\frac12.\end{array}\right.
\end{eqnarray}

\begin{remark}\label{rem:phase}
{\em An inspection of the above argument, in  combination with our
algorithm that was presented in Section~\ref{sec:alg}, shows that
the phase transition in~\eqref{eq:phase} at $c=\frac12$ corresponds
to a qualitative change in the optimal algorithm: after shifting the
vectors $\{v_1,\ldots,v_k\}$ so that $w(B)=0$ and renormalizing by
$R(B)$, for $c>\frac12$ the algorithm projects the points obtained
from the SDP to $\R^2$ and classifies them according to a partition
of $\R^2$ into three cones of positive measure, while for
$c<\frac12$ the partitioning is into two half-planes and the third
set (the one weighted by $c$) is empty.}
\end{remark}

\bibliographystyle{abbrv}

\bibliography{smola}

\end{document}